\documentclass[a4paper,onecolumn, unpublished]{compositionalityarticle}
\pdfoutput=1

\usepackage[utf8]{inputenc}
\usepackage{amsmath}
\usepackage{amssymb}
\usepackage{bbold}
\usepackage{graphicx}
\usepackage{scalerel}
\usepackage{braket}
\usepackage{leftidx}
\usepackage[title]{appendix}
\usepackage{tikz}
\usepackage{tikzit}
\usepackage{tikz-cd}
\usepackage{hyperref} 
\usepackage{cite}
\usepackage{geometry}
\usepackage{mathtools}
\usepackage{ stmaryrd }
\usepackage{circuitikz}
\usepackage{textcomp}
\usepackage{accents}
\usepackage{comment}
\usepackage{ mathrsfs }
\usepackage{authblk}
\usepackage{breakurl,amsmath,amsthm,amssymb,xspace,xypic,enumerate,color}
\usepackage{float}
\usepackage[numbers, square]{natbib}
\usepackage{hyperref}

\usepackage{comment}

\tikzstyle{doubled}=[line width=1.5pt] % set the line width for all doubled (quantum) maps/wires

\tikzstyle{dot}=[inner sep=0mm,minimum width=2mm,minimum height=2mm,draw,shape=circle]  
\tikzstyle{ddot}=[inner sep=0mm, doubled, minimum width=2.5mm,minimum height=2.5mm,draw,shape=circle]

\tikzstyle{pdot}=[inner sep=0mm, doubled, minimum width=2.5mm,minimum height=2.5mm,shape=circle]
\tikzstyle{phase dimensions}=[minimum size=6mm,font=\footnotesize,inner sep=0.2mm,outer sep=-2mm]

\tikzstyle{phase dot}=[pdot,phase dimensions]
\tikzstyle{wphase dot}=[dot, phase dimensions]

\tikzstyle{hadamard}=[fill=white,draw,inner sep=0.6mm,font=\footnotesize,minimum height=6mm,minimum width=8mm]

\tikzstyle{anti} = [fill=white,draw,inner sep=0.6mm,font=\footnotesize,minimum height=3mm,minimum width=3mm]

\tikzstyle{triang}=[regular polygon,regular polygon sides=3,draw,scale=0.75,inner sep=-0.75pt,minimum width=9mm,fill=white,regular polygon rotate=180]
\tikzstyle{triang_lesssep}=[regular polygon,regular polygon sides=3,draw,scale=0.75,inner sep=-4pt,minimum width=9mm,fill=white,regular polygon rotate=180, text depth=4mm]
\tikzstyle{triangdag}=[regular polygon,regular polygon sides=3,draw,scale=0.75,inner sep=-0.5pt,minimum width=9mm,fill=white]

\makeatletter
\newcommand{\boxshape}[3]{%
\pgfdeclareshape{#1}{
\inheritsavedanchors[from=rectangle] % this is nearly a rectangle
\inheritanchorborder[from=rectangle]
\inheritanchor[from=rectangle]{center}
\inheritanchor[from=rectangle]{north}
\inheritanchor[from=rectangle]{south}
\inheritanchor[from=rectangle]{west}
\inheritanchor[from=rectangle]{east}
\backgroundpath{% this is new
\southwest \pgf@xa=\pgf@x \pgf@ya=\pgf@y
\northeast \pgf@xb=\pgf@x \pgf@yb=\pgf@y

\@tempdima=#2
\@tempdimb=#3

\pgfpathmoveto{\pgfpoint{\pgf@xa - 5pt + \@tempdima}{\pgf@ya}}
\pgfpathlineto{\pgfpoint{\pgf@xa - 5pt - \@tempdima}{\pgf@yb}}
\pgfpathlineto{\pgfpoint{\pgf@xb + 5pt + \@tempdimb}{\pgf@yb}}
\pgfpathlineto{\pgfpoint{\pgf@xb + 5pt - \@tempdimb}{\pgf@ya}}
\pgfpathlineto{\pgfpoint{\pgf@xa - 5pt + \@tempdima}{\pgf@ya}}
\pgfpathclose
}
}}

\boxshape{NEbox}{0pt}{5pt}
\boxshape{SEbox}{0pt}{-5pt}
\boxshape{NWbox}{5pt}{0pt}
\boxshape{SWbox}{-5pt}{0pt}
\boxshape{EBox}{-3pt}{3pt}
\boxshape{WBox}{3pt}{-3pt}
\makeatother

\tikzstyle{map}=[draw,shape=NEbox,inner sep=2pt,minimum height=6mm,fill=white]
\tikzstyle{mapdag}=[draw,shape=SEbox,inner sep=2pt,minimum height=6mm,fill=white]
\tikzstyle{maptrans}=[draw,shape=SWbox,inner sep=2pt,minimum height=6mm,fill=white]
\tikzstyle{mapconj}=[draw,shape=NWbox,inner sep=2pt,minimum height=6mm,fill=white]

\tikzstyle{dmap}=[draw,doubled,shape=NEbox,inner sep=2pt,minimum height=6mm,fill=white]
\tikzstyle{dmapdag}=[draw,doubled,shape=SEbox,inner sep=2pt,minimum height=6mm,fill=white]
\tikzstyle{dmaptrans}=[draw,doubled,shape=SWbox,inner sep=2pt,minimum height=6mm,fill=white]
\tikzstyle{dmapconj}=[draw,doubled,shape=NWbox,inner sep=2pt,minimum height=6mm,fill=white]

\makeatletter
\pgfdeclareshape{cornerpoint}{
\inheritsavedanchors[from=rectangle] % this is nearly a rectangle
\inheritanchorborder[from=rectangle]
\inheritanchor[from=rectangle]{center}
\inheritanchor[from=rectangle]{north}
\inheritanchor[from=rectangle]{south}
\inheritanchor[from=rectangle]{west}
\inheritanchor[from=rectangle]{east}
\backgroundpath{% this is new
\southwest \pgf@xa=\pgf@x \pgf@ya=\pgf@y
\northeast \pgf@xb=\pgf@x \pgf@yb=\pgf@y

\pgfmathsetmacro{\pgf@shorten@left}{\pgfkeysvalueof{/tikz/shorten left}}
\pgfmathsetmacro{\pgf@shorten@right}{\pgfkeysvalueof{/tikz/shorten right}}

\pgfpathmoveto{\pgfpoint{0.5 * (\pgf@xa + \pgf@xb)}{\pgf@ya - 5pt}}
\pgfpathlineto{\pgfpoint{\pgf@xa - 8pt + \pgf@shorten@left}{\pgf@yb - 1.5 * \pgf@shorten@left}}
\pgfpathlineto{\pgfpoint{\pgf@xa - 8pt + \pgf@shorten@left}{\pgf@yb}}
\pgfpathlineto{\pgfpoint{\pgf@xb + 8pt - \pgf@shorten@right}{\pgf@yb}}
\pgfpathlineto{\pgfpoint{\pgf@xb + 8pt - \pgf@shorten@right}{\pgf@yb - 1.5 * \pgf@shorten@right}}
\pgfpathclose
}
}

\pgfdeclareshape{cornercopoint}{
\inheritsavedanchors[from=rectangle] % this is nearly a rectangle
\inheritanchorborder[from=rectangle]
\inheritanchor[from=rectangle]{center}
\inheritanchor[from=rectangle]{north}
\inheritanchor[from=rectangle]{south}
\inheritanchor[from=rectangle]{west}
\inheritanchor[from=rectangle]{east}
\backgroundpath{% this is new
\southwest \pgf@xa=\pgf@x \pgf@ya=\pgf@y
\northeast \pgf@xb=\pgf@x \pgf@yb=\pgf@y

\pgfmathsetmacro{\pgf@shorten@left}{\pgfkeysvalueof{/tikz/shorten left}}
\pgfmathsetmacro{\pgf@shorten@right}{\pgfkeysvalueof{/tikz/shorten right}}

\pgfpathmoveto{\pgfpoint{0.5 * (\pgf@xa + \pgf@xb)}{\pgf@yb + 5pt}}
\pgfpathlineto{\pgfpoint{\pgf@xa - 8pt + \pgf@shorten@left}{\pgf@ya + 1.5 * \pgf@shorten@left}}
\pgfpathlineto{\pgfpoint{\pgf@xa - 8pt + \pgf@shorten@left}{\pgf@ya}}
\pgfpathlineto{\pgfpoint{\pgf@xb + 8pt - \pgf@shorten@right}{\pgf@ya}}
\pgfpathlineto{\pgfpoint{\pgf@xb + 8pt - \pgf@shorten@right}{\pgf@ya + 1.5 * \pgf@shorten@right}}
\pgfpathclose
}
}

\makeatother

\pgfkeyssetvalue{/tikz/shorten left}{0pt}
\pgfkeyssetvalue{/tikz/shorten right}{0pt}

\tikzstyle{kpoint common}=[draw,fill=white,inner sep=1pt,minimum height=4mm]
\tikzstyle{kpoint}=[shape=cornerpoint,shorten left=5pt,kpoint common]
\tikzstyle{kpoint adjoint}=[shape=cornercopoint,shorten left=5pt,kpoint common]
\tikzstyle{kpoint conjugate}=[shape=cornerpoint,shorten right=5pt,kpoint common]
\tikzstyle{kpoint transpose}=[shape=cornercopoint,shorten right=5pt,kpoint common]

\tikzstyle{kpointdag}=[kpoint adjoint]
\tikzstyle{kpointadj}=[kpoint adjoint]
\tikzstyle{kpointconj}=[kpoint conjugate]
\tikzstyle{kpointtrans}=[kpoint transpose]

\tikzstyle{big kpoint}=[kpoint, minimum width=1.0 cm, minimum height=2mm, inner sep=4pt, text depth=1.5mm]

 \tikzstyle{upground}=[circuit ee IEC,thick,ground,rotate=90,scale=1.5]
 \tikzstyle{downground}=[circuit ee IEC,thick,ground,rotate=-90,scale=1.5]

\tikzstyle{discarding}=[fill=white, draw=black, shape=circle, style=upground]
\tikzstyle{smalldiscarding}=[fill=white, draw=black, style=upground, scale=0.5]
\tikzstyle{backdiscard}=[fill=white, draw=black, shape=circle, style=downground, scale=0.5]
\tikzstyle{smallbackdiscard}=[fill=white, draw=black, shape=circle, style=downground, scale=0.5]
\tikzstyle{state}=[fill=white, draw=black, style=triang, tikzit shape=rectangle]
\tikzstyle{kstate}=[fill=white, draw=black, style=kpoint, tikzit shape=rectangle]
\tikzstyle{kstateconj}=[fill=white, draw=black, style=kpoint conjugate, tikzit shape=rectangle]
\tikzstyle{kstateBIG}=[fill=white, draw=black, style=big kpoint, tikzit shape=rectangle]
\tikzstyle{effect}=[fill=white, draw=black, style=triangdag]
\tikzstyle{keffect}=[fill=white, draw=black, style=kpoint adjoint]
\tikzstyle{keffectconj}=[fill=white, draw=black, style=kpoint transpose]
\tikzstyle{morphdag}=[style=mapdag]
\tikzstyle{morph}=[style=hadamard]
\tikzstyle{WIDEmorph}=[style=hadamard, minimum width=14mm]
\tikzstyle{morphtrans}=[style=maptrans]
\tikzstyle{morphconj}=[style=mapconj]
\tikzstyle{CPMmorph}=[style=dmap]
\tikzstyle{CPMmorphconj}=[style=dmapconj]
\tikzstyle{CPMmorphdag}=[style=dmapdag]
\tikzstyle{CPMmorphtrans}=[style=dmaptrans]
\tikzstyle{CPMstate}=[fill=white, draw=black, style=triang, doubled]
\tikzstyle{CPMstateBIG}=[fill=white, draw=black, style={triang_lesssep}, doubled]
\tikzstyle{CPMkstate}=[fill=white, draw=black, style=kpoint, tikzit shape=rectangle, doubled]
\tikzstyle{CPMkstateconj}=[fill=white, draw=black, style=kpoint conjugate, tikzit shape=rectangle, doubled]
\tikzstyle{CPMkstateBIG}=[fill=white, draw=black, style=big kpoint, tikzit shape=rectangle, doubled]
\tikzstyle{CPMkeffect}=[fill=white, draw=black, style=kpoint adjoint, doubled]
\tikzstyle{CPMkeffectconj}=[fill=white, draw=black, style=kpoint transpose, doubled]
\tikzstyle{UHfB}=[fill=white, draw=black, style=triangdag, doubled, inner sep=-2pt]
\tikzstyle{leak}=[style=tinypoint, regular polygon rotate=-90]
\tikzstyle{leakfill}=[style=tinypoint, regular polygon rotate=-90, fill=black]
\tikzstyle{Z}=[style=dot, fill=green]
\tikzstyle{X}=[style=dot, fill=red]
\tikzstyle{black_dot}=[style=dot, fill=black]
\tikzstyle{white_dot}=[style=dot, fill=white]
\tikzstyle{qblack_dot}=[style=ddot, fill=black]
\tikzstyle{qwhite_dot}=[style=ddot, fill=white]
\tikzstyle{whitephase}=[style=wphase dot, fill=white]
\tikzstyle{qredphase}=[style=phase dot, fill=red]
\tikzstyle{qgreenphase}=[style=phase dot, fill=green]
\tikzstyle{had}=[style=hadamard, doubled]
\tikzstyle{box}=[style=hadamard]
\tikzstyle{classhad}=[style=hadamard]
\tikzstyle{antipode}=[style=anti]

\tikzstyle{dottededge}=[-, dotted]
\tikzstyle{double edge}=[-, style=doubled, draw=black, tikzit draw={rgb,255: red,18; green,168; blue,191}]
\tikzstyle{new edge style 0}=[<-]
\tikzstyle{new edge style 1}=[-, draw={rgb,255: red,223; green,66; blue,126}, fill={rgb,255: red,223; green,66; blue,126}]
\tikzstyle{new edge style 2}=[-, draw={rgb,255: red,14; green,188; blue,83}]
\tikzstyle{new edge style 3}=[<-, draw={rgb,255: red,223; green,66; blue,126}]

\usetikzlibrary{calc, positioning, shapes.geometric}
\usetikzlibrary{
	arrows,
	shapes,
	decorations,
	intersections,
	backgrounds,
	positioning,
	circuits.ee.IEC
	}

\newcommand{\tikzfigscale}[2]{\scalebox{#1}{\input{#2.tikz}}}

\def\be{\begin{equation}}
\def\ee{\end{equation}}
\def\ba{\begin{align}}
\def\ea{\end{align}}

\usepackage{bbm}

\newcommand{\lalg}[1]{\mathfrak{#1}}  % Lie algebra

\renewcommand{\sl}{\lalg{sl}}

\newtheorem{definition}{Definition}
\newtheorem{theorem}{Theorem}
\newtheorem{corollary}{Corollary}
\newtheorem{lemma}{Lemma}

\newtheorem{example}{Example}

\tikzstyle{every picture}=[baseline=-0.25em,shorten <=-0.1pt]
\tikzstyle{dotpic}=[scale=0.5]
\tikzstyle{braceedge}=[decorate,decoration={brace,amplitude=1mm,raise=-1mm}]
\tikzstyle{dot}=[inner sep=0.7mm,minimum width=0pt,minimum height=0pt,fill=black,draw=black,shape=circle]
\tikzstyle{small dot}=[inner sep=0.1mm,minimum width=0pt,minimum height=0pt,fill=black,draw=black,shape=circle]
\tikzstyle{black dot}=[dot]
\tikzstyle{white dot}=[dot,fill=white]
\tikzstyle{gray dot}=[dot,fill=gray!40!white]
\tikzstyle{alt white dot}=[white dot,label={[xshift=3mm,yshift=-0.05mm,font=\tiny]left:$*$}]
\tikzstyle{alt gray dot}=[gray dot,label={[xshift=3mm,yshift=-0.05mm,font=\tiny]left:$*$}]
\tikzstyle{white norm}=[rectangle,fill=white,draw=black,minimum height=2mm,minimum width=2mm,inner sep=0pt,font=\small]
\tikzstyle{gray norm}=[white norm,fill=gray!40!white]
\tikzstyle{square box}=[rectangle,fill=white,draw=black,minimum height=5mm,minimum width=5mm,font=\small]
\tikzstyle{square gray box}=[rectangle,fill=gray!30,draw=black,minimum height=6mm,minimum width=6mm]
\tikzstyle{diredge}=[->]
\tikzstyle{rdiredge}=[<-]
\tikzstyle{dashed edge}=[dashed]
\tikzstyle{cross}=[preaction={draw=white, -, line width=3pt}]

\newcommand{\dotdualmult}[1]{%
\!\begin{tikzpicture}[dotpic]
    \node [style=white dot] (0) at (0, 0.3) {};
    \node [style=none] (1) at (-0.5, -0.4) {};
    \node [style=none] (2) at (0.5, -0.4) {};
    \node [style=none] (3) at (0, 0.8) {};
    \draw [style=diredge] (3.center) to (0);
    \draw [style=diredge, in=15, out=-30, looseness=1.50] (0) to (1.center);
    \draw [style=diredge, in=165, out=-150, looseness=1.50] (0) to (2.center);
\end{tikzpicture}\!}

\newcommand{\dotconorm}[1]{%
\,\begin{tikzpicture}[dotpic,yshift=0.4mm]
    \node [style=none] (0) at (0, -0.4) {};
    \node [style=white norm] (1) at (0, 0.1) {};
    \node [style=none] (2) at (0, 0.5) {};
    \draw [style=diredge] (1) to (0.center);
    \draw (2.center) to (1);
\end{tikzpicture}\,}

\newcommand{\astfootnote}[1]{
\let\oldthefootnote=\thefootnote
\setcounter{footnote}{0}
\renewcommand{\thefootnote}{\fnsymbol{footnote}}
\footnote{#1}
\let\thefootnote=\oldthefootnote
}

\title{A Mathematical Framework for Transformations of Physical Processes}
\author{Matt Wilson}
\email{matthew.wilson@cs.ox.ac.uk}
\affiliation{Quantum Group, Department of Computer Science, University of Oxford}
\affiliation{HKU-Oxford Joint Laboratory for Quantum Information and Computation}
	
\author{Giulio Chiribella}
\email{giulio.chiribella@cs.ox.ac.uk}
\affiliation{QICI Quantum Information and Computation Initiative, Department of Computer Science}
\affiliation{Quantum Group, Department of Computer Science, University of Oxford}
\affiliation{HKU-Oxford Joint Laboratory for Quantum Information and Computation}
\affiliation{Perimeter Institute for Theoretical Physics, 31 Caroline Street North, Waterloo, Ontario, Canada}

\begin{document} \emergencystretch 3em

\maketitle

\begin{abstract}

We observe that the existence of sequential and parallel composition supermaps in higher order theories of transformations can be formalised using enriched category theory. Encouraged by relevant examples such as unitary supermaps and layers within higher order causal categories (HOCCs), we treat the modelling of higher order physical theories with enriched monoidal categories in analogy with the modelling of physical theories with monoidal categories. We use the enriched monoidal setting to construct a suitable definition of structure preserving map between higher order physical theories via the Grothendieck construction. We then show that the convenient feature of currying in higher order physical theories can be seen as a consequence of combining the primitive assumption of the existence of parallel and sequential composition supermaps with an additional feature of \textit{linking}. We then use our definition of structure preserving map to show that categories containing infinite towers of enriched monoidal categories with full and faithful structure preserving maps between them inevitably lead to closed monoidal structures. The aim of the proposed definitions is to step towards providing a broad framework for the study and comparison of novel causal structures in quantum theory, and, more broadly, a paradigm of physical theory where static and dynamical features are treated in a unified way.  %We proceed to define the notion of a theory which permits the static manipulations of all of its dynamic processes, showing that when equipped with a particular well motivated isomorphism they in fact co-incide with closed monoidal categories.
%We then show that characterise an easy to state subclass of closed symmetric monoidal categories in terms of a basic interaction between lower and higher order theories. 
%We then use the framework to present a theorem on the inevitability of closed monoidal structure for theories which contain infinite towers of super-physical theories, each layer being a theory of manipulations of the processes within the layer beneath it. 

\end{abstract}

\tableofcontents

\section{Introduction}
Traditionally,  physical theories have been concerned with the laws governing the evolution of certain physical systems, such as particles or fields. %Physical systems are the primary objects,  and the purpose of a theory is to describe the possible processes that the systems of interest can undergo. 
 In the ontology of a theory, the physical systems are regarded as fundamental objects, while their evolution  is  regarded as a tool for predicting relations among objects in different regions of space and time.   
   Over the past decade,  a series of works in quantum information theory started exploring the idea that processes themselves could be regarded as objects, which can be acted upon by a kind of higher order physical transformations, known as quantum supermaps   \cite{Chiribella2008TransformingSupermaps,chiribella2009theoretical,Chiribella2013QuantumStructure,chiribella2013normal, perinotti2017causal,Bisio2019TheoreticalTheory}.  %The simplest example of a quantum supermap is the insertion of a quantum gate  into a quantum circuit, which effectively results into a new quantum gate.   More elaborate examples involve the combination of multiple processes in a variety of configurations.   
The introduction of the quantum supermap framework has led to developments in quantum information theory by giving a way to formulate protocols in which transformations are treated as resources,  \cite{chiribella2008optimal,bisio2009optimal,bisio2010optimal,Ebler2018EnhancedOrder,gour2019quantify,miyazaki2019complex,quintino2019probabilistic,sedlak2019optimal, Guerin2016ExponentialCommunication}, and has led to the study of new types of causal structures compatible with quantum mechanics \cite{Chiribella2013QuantumStructure,Oreshkov2012QuantumOrder,Branciard2015TheViolation, device_independent_causal_lugt, gogioso2023geometry}. In addition, higher order transformations  provide a broad framework for general physical theories with dynamical causal structure, and, as originally suggested by Hardy \cite{hardy2007towards} are expected contribute to the formulation of a complete theory of quantum gravity. Complementary to this research direction is the development of programming languages which permit higher order types whilst retaining compatibility with quantum theory (by forbidding cloning \cite{Wootters1982ACloned}, the signature of the Cartesian monoidal structure underlying the standard lambda calculus), such as \textit{linear} or \textit{quantum} lambda calculi  \cite{Selinger2004TowardsLanguage, Selinger2005AControl, VanTonder2004AComputation, BentonLinear-CalculusRevisited, Pagani2013Applying, Ambler1991FirstCategories}.
  %Kissinger2019AStructure}.

A compositional foundation for the study of physical theories, including quantum and classical theory, is provided by the process theory framework \cite{Coecke2017PicturingReasoning}. The framework is built on the notion of a symmetric monoidal category, which captures some basic structures present in a broad class of physical theories.  Such structures include a notion of system, a notion of processes between systems, and, crucially, a notion of the sequential and parallel composition of processes, diagrammatically represented as
\begin{equation}
  \input{figs/proccomp.tikz}.
 \end{equation}
 The tensor unit $I$ represents the {\em trivial system}, and, for every object $A$,  the processes of type $I \to A$ are viewed as {\em states}  of system $A$. 
 The process theoretic treatment has contributed new insights into quantum foundations and the general structure of physical theories \cite{Selby2018ReconstructingPostulates, Vicary2011CategoricalAlgebras, Coecke2007QuantumSums,Coecke2009InteractingDiagrammatics,Coecke2010TheEntanglement, Coecke2011PhaseQubits, Coecke2013CausalProcesses,  Heunen2013IntroductionMechanics, Selinger2007DaggerAbstract, Tull2020ATHEORY, GalleyATheory}, leading to the natural question of whether new insights could be gained in the recently developed higher order approach to the foundations of physics by developing an analogous framework of \textit{higher order process theories}.  %In addition to the above formal circuit based representation for quantum processes, supermaps are often informally represented by boxes with open holes, into which processes may be inserted: \begin{equation}
%  \tikzfig{figs/informal}
% \end{equation}

Recently, a variety of works set out to establish compositional features of higher order processes, in linear algebraic \cite{Bisio2019TheoreticalTheory, apadula2022nosignalling, hoffreumon2022projective}, and categorical \cite{Kissinger2019AStructure,pbs_paper, simmons2022higherorder} settings. Most relevant to this paper, in \cite{Kissinger2019AStructure}, a process theoretic framework for supermaps was developed for the purpose of providing a categorical language for causal structures. In this framework, causal structures are represented by the objects of a $*$-autonomous category of higher order processes $\mathbf{Caus(}\mathbf{C}\mathbf{)}$ built from a ``pre-causal'' category $\mathbf{C}$. This construction reveals deep relations between $*$-autonomy and the structure of higher order transformations in quantum theory, in particular producing a convenient type system for reasoning about causal structures. However, there is a sense that the notion of a raw-material pre-causal (and so compact closed) category may be too restrictive a requirement in the study of infinite dimensional systems such as those encountered in quantum field theory, and ultimately quantum theories of gravity.

In this work we aim to to pin down a notion of a higher order theory as a mathematical structure in its own right, independently of the study of causality, and independently of the notion of a raw material category from which a theory might be constructed. Our main motivations for formalisation of structural features of higher order physical theories are the following:

\begin{itemize}
    \item Quantum resource theories are often modelled in the abstract by construction from the notion of a sub-symmetric monoidal category \cite{Coecke2014AResources}, at the level of identifying quantum theory as forming a symmetric monoidal category, this essentially phrases the definition of a resource theory as simply the definition of a sub-theory. To build on the approaches of \cite{takagiresourcechannels, Kristjansson2020ResourceCommunication, gour2020dynamical} and fully extend resource theories in a satisfactory way to higher order quantum theory, we need to develop a notion of higher order sub-theory.
    \item Identification of basic compositional features of higher order physics could open the door for the study of higher order principles for the axiomatisation of physical theories, in analogy with reconstructions for standard quantum theory which work with symmetric monoidal categories as a background axiom \cite{Selby2018ReconstructingPostulates, Tull2020ATHEORY, Chiribella_2010_prob_recon}. Preliminary investigations in this direction, which build from the results of this paper by beginning from the assumption of closed monoidal structure and then imposing causality principles, are presented in \cite{wilsoncausality}.
    \item Once enough structural features are identified it may be possible to characterise supermap definitions in terms of universal properties. The generally accepted definition of first-order quantum process can be motivated in this way, as arising from a universal property with respect to affine monoidal structure \cite{statonDBLP:journals/corr/abs-1901-10117}\footnote{In contrast to the more stable state of affairs in standard quantum theory, the identification of certain constructions of supermaps as universal could offer an alternative perspective on the still-open issue \cite{Ara_jo_2017_purification, Feix_2016_causal_model} of which of the numerous classes of supermaps \cite{Wechs_2021_control} should be considered reasonable.}.
\end{itemize}

The contributions of this paper can be summarised in three parts, first, enrichment is used to formalise in categorical terms the existence of sequential and parallel composition supermaps. Second, suitable structure preserving maps with respect to enrichment are defined. Third, closed monoidal structure of higher order quantum theories is motivated in terms of extra axioms on top of monoidal enrichment, in doing so the previously developed notion of structure-preserving map is applied. We now expand on each of these points in more detail. 
\paragraph{Process Manipulation:}
Any symmetric monoidal category $\mathbf{C}$, can be interpreted as theory of processes which can be composed in sequence or in parallel. A feature common to theories of transformations of processes is the existence of higher order transformations which actually \textit{perform} these composition rules. Namely, a key feature is the existence of higher order processes which put processes together in sequence or in parallel \[  \input{figs/axiom_seq_a.tikz} \quad \quad \quad \quad \input{figs/axiom_par_1.tikz}. \] This feature of higher order processes can be identified with the categorical definition of a $\mathbf{V}_{\cong}$-smc $\mathbf{C}$. This equips the higher theory $\mathbf{V}$ with types $[A,A']$ representing the space of processes of type $A \rightarrow A'$ in a lower order theory $\mathbf{C}$, and a tensor $\otimes_{V}$ so that $[A,A'] \otimes [B,B']$ represents the space of bipartite processes which can be plugged together in either order, being treated as freely manipulable.
 \paragraph{Structure Preserving Maps:}
 Our second contribution is to show that one can define structure preserving maps between theories of higher order transformations. These are defined to be functors $\mathcal{F}^C:\mathbf{C} \rightarrow \mathbf{C}', \mathcal{F}^{V}:\mathbf{V} \rightarrow \mathbf{V}'$ on the lower and higher order parts of the theory along with morphisms of type $\mathcal{F}_{AB}:\mathcal{F}^{V}[A,B] \rightarrow [\mathcal{F}^{C}A,\mathcal{F}^{C}B]$ which encode the preservation  of enriching structure \[  \input{figs/inuitive_functor.tikz}. \] This introduced notion of structure preserving map allows us to formalise the following, given the statements $(a)$: $\mathbf{C}^3$ is a theory of higher order transformations for $\mathbf{C}^2$ and $(b)$: $\mathbf{C}^2$ is a theory of higher order transformations for $\mathbf{C}^1$, then there is a morphism $\Gamma:(b) \rightarrow (a)$ embedding the latter statement into the former \[  \input{figs/inuitive_functor_2.tikz}. \] A key feature of such structure preserving maps is that they are composable, forming a category of higher-order theories. 
 \paragraph{Linking and Process Manipulation $\rightarrow$ Closed Monoidal Structure:}
 The third contribution is to show that the possibility to \textit{curry} processes, that is, the existence of closed monoidal structure, can be derived from appending an additional notion to the above primitive operational principles observed in higher order physical theories. We show that currying can be viewed as a consequence of (i) The possibility to compose processes in sequence, (ii) The possibility to compose processes in parallel, (iii) The possibility to translate between an object $A$ and the space $[I,A]$ of states on $A$. These principles are combined together in the definition of a \textit{linked}, and \textit{faithful}, enriched monoidal category, which is shown to be equivalent to the definition of closed monoidal category. The crux of the proof can be conveyed intuitively using the following picture \[  \input{figs/axiom_eval.tikz}.   \] Here we see links used to convert between first and second-order systems, combined with sequential composition supermaps to construct an \textit{evaluation} process of type $A \otimes [A,B] \rightarrow B$. This result allows us to state, a series of reasonable physical principles which motivate working with closed monoidal structure without \textit{directly} assuming the possibility to curry processes. 
 
 By using the introduced notion of structure  preserving map $\Gamma:(b) \rightarrow (a)$ between layers within sequences of higher order theories: we proceed to generalise this result by showing that any infinite sequence of enriched monoidal categories, with well behaved structure preserving maps between the layers of the sequence, leads to closed monoidal structure \[\input{figs/hopt_2a.tikz}. \] Closed monoidal structure is a simple to state and easy to interpret mathematical structure on top of monoidal structure, the aim of these results is to show that deductions made by combining other physical principles with closed monoidal structure are likely to be general statements about higher order physical theories, first steps in the direction of combining closed monoidal structure with other standard physical principles such as causality and determinism are given in \cite{wilsoncausality}.

\section{Notation and basic definitions}
We will use the abbreviation SMC for \textit{symmetric monoidal category}, the definition of which may be found in \cite{Lane1971CategoriesMathematician}. An SMC consists of objects $A,B, \dots$ morphisms $f:A \rightarrow B$ and composition rules $(\circ, \otimes)$. A morphism $f:A \rightarrow B$ can be represented by a box with input wire $A$ and output wire $B$. The parallel composition $f \otimes g$ of morphisms is written by placing $f$ next to $g$, the sequential composition $g \circ f$ of $f:A \rightarrow B$ and $g:B \rightarrow C$ is written by connecting boxes along wire $B$ as in the following pictures: 
\begin{equation}
  \input{figs/proccomp.tikz} .
 \end{equation}
There is furthermore a unit object $I$ which is not explicitly written, interpreted as representing only empty space. Similarly associativity of sequential composition and associativity of parallel composition up to natural isomorphism are absorbed into the graphical language, neither being explicitly written. The categorical notion of one monoidal category living inside another is that of a \textit{monoidal functor}. In general, there may be more than one way of  representing a monoidal category $\mathbf{C}$ inside another monoidal category $\mathcal{D}$. The notion of equivalence between two representations  is that of a \textit{monoidal natural isomorphism}. We will often refer to a full subcategory of a category $\mathbf{C}$ with objects given by combining all objects in some collection $S \subseteq ob(\mathbf{C})$ iteratively using some family of functions $\boxtimes_k :ob(\mathbf{C}) \times ob(\mathbf{C}) \rightarrow ob(\mathbf{C})$. As a shorthand for such a collection generated by $S$ and functions $\boxtimes_k$ we use the symbol $S|\boxtimes_1| \dots | \boxtimes_n$.

\paragraph{Closed Monoidal Categories:}
 Closed monoidal structure \cite{Lane1971CategoriesMathematician} is a standard categorical structure defined with the purpose of abstracting the notion of currying found in the category $\mathbf{Set}$ of functions between sets. Currying for functions is the property that for every pair of sets $A,B$ there exists a function $\texttt{eval}_{A,B}: A \times \mathbf{Set}(A,B) \rightarrow B$ which on elements is defined by $\texttt{eval}_{A,B}(a,f) : = f(a)$. In the general monoidal setting then closed monoidal structure of a monoidal category $\mathbf{C}$ is given by requiring a co-universal arrow of type $\texttt{eval}_{A,B}:A \otimes [A,B] \rightarrow B$. Explicitly such a co-universal arrow specifies for each morphism $f:A \otimes C \rightarrow B$ a unique morphism $\hat{f}:C \rightarrow [A,B]$ such that \[\input{figs/closed_monoidal.tikz}. \] This generalises currying by giving a natural isomorphism between $\mathbf{C}(A \otimes  C,B)$ and $\mathbf{C}(C,[A,B])$. The evaluation morphism can intuitively be understood as an open hole into which a process can be inserted.

\paragraph{Higher Order Causal Categories:}
Moving beyond first-order process theories in which objects are interpreted as representing state spaces and processes are interpreted as transformations of states, is the higher order framework for quantum theory, in which transformations of processes are considered, intuitively represented as \begin{equation}
  \input{figs/informal.tikz}.
 \end{equation} To define iterated higher order transformations of quantum processes a construction of higher order quantum theory ($\mathbf{HOQT}$) is provided in \cite{Bisio2019TheoreticalTheory} which takes advantage of the Choi isomorphism. The Choi isomorphism can be viewed as a particular consequence of compact closure, and using this observation the deterministic part of $\mathbf{HOQT}$ can be generalised to a wider variety or raw-material physical theories. For any pre-causal (and so compact closed) category $\mathbf{C}$ a category $\mathbf{Caus}[\mathbf{C}]$, which we now for notational convenience refer to as $[\mathbf{C}]$, can be constructed \cite{Kissinger2019AStructure}. $[\mathbf{C}]$ includes lower order types $A,B, \dots$, two tensor products $(\otimes, \& )$, and methods for constructing higher order types given by closed monoidal structure. To each pair of objects $A,B$ an object $[A,B]$ is specified representing the space of transformations from $A$ to $B$. Higher order processes can then be represented as those acting on higher order types, for instance $S:[A,A'] \rightarrow [B,B']$ represents a quantum supermap from processes of type $A \rightarrow A'$ to processes of type $B \rightarrow B'$. In this paper we will refer to the subcategory $[\mathbf{C}]^1$ as the category with objects given by \textit{first order types} as defined in \cite{Kissinger2019AStructure}. For $\mathbf{C} := \mathbf{CPM}[\mathbf{FHilb}]$ then $[\mathbf{C}]^1$ is equivalent to the category of CPTP quantum processes. We similarly refer to objects in $[A_1,B_1]|\otimes| \& $ with $A_1,B_1$ first order types as second order types and define $[\mathbf{C}]^{2}$ to be the full subcategory with objects given by second order types.  Iterating this we can define $[\mathbf{C}]^{n+1}$ with objects given by $[A_n,B_n]|\otimes| \& $. Each such category is symmetric monoidal with unit given by $I_{n} := [I,[I,[ \dots ]]$. 

 Intuitively, each of these categories $[\mathbf{C}]^{n+1}$ can be interpreted as a theory of higher order transformations of the processes in $[\mathbf{C}]^{n}$. In this paper we will argue that a key feature of $[\mathbf{C}]^{n+1}$ which allows it to be considered in this way is that it contains morphisms which implement the sequential and parallel composition of morphisms of $[\mathbf{C}]^{n}$. Let us now see how such morphisms can be seen to be present. For first order types $A,B,A',B'$ then $[A , A'] \otimes [B , B']$ represents the space of no-signalling channels from $AB$ to $A'B'$. I.E those which forbid signalling from $A$ to $B'$ and from $B$ to $A'$: \[ \input{figs/loop_sig_pre.tikz}.  \]
 
 For any $A,B,C$ in $[\mathbf{C}]^1$ there exists a morphism $\circ_{ABC}: [A,B] \otimes [B,C] \rightarrow [A,C]$ in $[\mathbf{C}]^2$ given by plugging together wires of non-signalling channels  \[ \input{figs/loop_sig.tikz}.  \] Indeed, one can see that $\circ_{ABC}$ is a morphism by checking it preserves first order processes, this is verified by noting that every deterministic non-signalling transformation factorises as: \[ \input{figs/loop_decomp.tikz}  \] and so applying the morphism $\circ_{ABC}$ gives: \[ \input{figs/loop_comp.tikz}  \] which is a morphism of $[\mathbf{C}]^1$ since $[\mathbf{C}]^1$ is an smc. This observation, that there is a composition process which can be applied to tensor products of types generalises to $[\mathbf{C}]^i$ and can be seen as a consequence of the fact that $[\mathbf{C}]$ is a closed monoidal category. In this sense the existence of a sequential composition process can be seen as a generalisation of the non-signalling property of quantum channels to higher order processes.

\section{Main Observation: Monoidal Enrichment}
In this section we aim to highlight and formalise a prominent feature of higher order physical theories, the existence of primitive higher order sequential and parallel composition processes. In short, we make the following observation \[ \textit{Higher Order Theories of Transformations are Enriched Symmetric Monoidal Categories}. \] The use of enrichment as a semantics for higher-order manipulation of functions has been previously observed in \cite{statonlmcs:6192}, where the higher order theory was taken to be Cartesian, we will however be interested in theories with non-trivial correlations between processes being manipulated such as those present between the two halves of a non-signalling channel. When we use the term ''enriched monoidal category" we mean a slight generalisation of it's standard usage \cite{Kelly2005}, rather than the notion of a $\mathbf{V}$-enriched symmetric monoidal category $\mathbf{C}$ (from now on termed $\mathbf{V}$-smc $\mathbf{C}$ for short) we use the notion of a $\mathbf{V}_{\cong}$- smc $\mathbf{C}$ (defined explicitly in Appendix A). In a $\mathbf{V}_{\cong}$- enriched symmetric monoidal category $\mathbf{C}$, for each pair of objects $A,B$ of $\mathbf{C}$ there exists an object $[A,B]$ in $\mathbf{V}$ whose states represent processes in the underlying category $\mathbf{C}$ via a bijection \[ \kappa: \mathbf{C}(A,B) \cong \mathbf{V}(I,[A,B]) . \] In string diagrams of $\mathbf{V}$ and $\mathbf{C}$ respectively the isomorphism can be represented by: \begin{equation}
    \input{figs/kappa.tikz}.
\end{equation}
From now on we refrain from explicitly writing $\kappa$ whenever its presence is clear. In the standard definition of a $\mathbf{V}$-smc $\mathbf{C}$ this bijection would be required to be an equality and so would not so easily incorporate standard constructions of higher order physical theories in which variants of the Choi isomorphism are used \cite{Choi1975CompletelyMatrices, Bisio2019TheoreticalTheory, Kissinger2019AStructure}. For each $[A,B]$ and $[B,C]$ in a $\mathbf{V}_{\cong}$-smc $\mathbf{C}$ there is a morphism in $\mathbf{V}$ which allows to plug their underlying processes together:  \[\bigcirc: [A,B] \otimes [B,C] \rightarrow [A,C]\] represented formally on the left as a string diagram in $\mathbf{V}$ and informally on the right to show intuitively its action on the underlying category $\mathbf{C}$:
\begin{equation}
    \input{figs/axiom_seq_a.tikz}.
\end{equation}
Associativity and unitality of the sequential composition process in $\mathbf{V}$ and the guarantee that it actually implements sequential composition for $\mathbf{C}$, are represented by: 
\begin{equation}
    \input{figs/comp_new.tikz}.
\end{equation}
Monoidal enrichment \cite{Kelly2005} provides furthermore a parallel composition process in $\mathbf{V}$: of type \[\otimes_{ABA'B'} : [A,A'] \otimes [B,B'] \rightarrow [A \otimes B, A' \otimes B']\] which can be represented formally and intuitively respectively by:
\begin{equation}
    \input{figs/axiom_par_1.tikz}.
\end{equation}
Again in a $\mathbf{V}_{\cong}$-smc $\mathbf{C}$ the following conditions are required, guaranteeing that $\otimes_{ABA'B'}$ really does behave like a parallel composition process:
\begin{equation}
    \input{figs/comppar_new.tikz}.
\end{equation}
These conditions represent associativity of parallel composition, parallel composition with empty space having no effect, and finally that the morphism indeed implements the parallel composition of processes respectively. Compatibility between symmetries is enforced by:
\begin{equation}
    \input{figs/symy_law.tikz},
\end{equation}
and lastly the condition:
\begin{equation}
    \input{figs/interchange.tikz},
\end{equation}
is required, which represents the interchange law between sequential and parallel composition. The above conditions are technically not well typed unless $\mathbf{C}$ is strictly monoidal, the relaxation to the case in which $\mathbf{C}$ is non-strict is given in Appendix A.

\begin{comment}
\paragraph{Full extend-ability}
Given two supermaps one ought to be able to apply them to \textit{any} bipartite process. \[  \input{figs/axiom_5_new.tikz} \] We model this by requiring an invertible natural transformation $p:[A \otimes B, A' \otimes B'] \rightarrow [A,A'] \boxtimes [B,B']$ so that any supermaps $S:[A,A'] \rightarrow [B,B']$ and $T:[C,C'] \rightarrow [D,D']$ can be combined  using $p^{-1} \circ (S \boxtimes T) \circ p$ an action of type $[A \otimes C, A' \otimes C'] \rightarrow [B \otimes D,B' \otimes D']$.

An example of a category $\mathbf{V}$ which allows for process manipulation on the processes of $\mathbf{C}$ but which does not clearly allow for full-extend-ability on $\mathbf{C}$ is the category $\mathbf{Set}$. Indeed given that $\mathbf{C}$ is a monoidal category this by definition entails that $\mathbf{C}$ is enriched in $\mathbf{Set}$. However it is not the case that we would expect any function $S:\mathbf{C(A,A')} \rightarrow \mathbf{C}(B,B')$ to represent a valid supermap, and indeed this is in part because it is not clear how to extend any function of thee above type to a function $\mathbf{C}(A \otimes E,A' o\times E') \rightarrow \mathbf{C}(B \otimes E, B' \otimes E')$.

\end{comment}

\subsection{Examples}
We now present a series of examples of theories of supermaps, and observe that each is indeed an example of an enriched symmetric monoidal category.
%\begin{example}[Higher Order Quantum Theory]
%For any $i$ there is a $\mathbf{HOQT}_{\cong}$-smc $\mathbf{HOQT}^{i}$ and similarly for any $i,i+1$ there is a $\mathbf{HOQT}^{i+1}_{\cong}$-smc $\mathbf{HOQT}^{i}$. In each case enrichment is ensured by the closed monoidal structure of $\mathbf{HOQT}$ \cite{Bisio2019TheoreticalTheory}.
%\end{example}

\begin{example}[Higher Order Causal Categories]
For any $i$ there is a $[\mathbf{C}]_{\cong}$-smc $[\mathbf{C}]^{i}$ and similarly for any $i,i+1$ there is a $[\mathbf{C}]^{i+1}_{\cong}$-smc $[\mathbf{C}]^{i}$. In each case enrichment is ensured by the closed monoidal structure of $[\mathbf{C}]$ \cite{Kissinger2019AStructure}. %The same can be said for $[\mathbf{C}_{\leq}]$ and each $[\mathbf{C}_{\leq}]^{i}$.
\end{example}

\begin{example}[Superunitaries]
For any sub symmetric monoidal category $\mathbf{M} \subseteq [\mathbf{C}]^{1}$ (where we consider $[\mathbf{C}]$ equipped with the $\otimes$ product) such that \[ \forall \phi \in \mathbf{M} \quad s.t \quad \input{figs/subcat_pre_1.tikz}: \quad \quad  \exists L,M \in \mathbf{M} \quad s.t  \quad \input{figs/subcat_pre_2.tikz} ,  \] one can construct a category $\mathbf{M}^2$ of completely-$\mathbf{M}$-preserving supermaps and a corresponding $\mathbf{M}^2_{\cong}$-smc $\mathbf{M}$. The objects of $\mathbf{M}^2$ are taken to be objects in $[\mathbf{M},\mathbf{M}] | \otimes | \&$\footnote{Here we use $[\mathbf{M},\mathbf{M}]$ to represent the collection of all objects of the form $[A,B]$ with $A$ and $B$ objects of $\mathbf{M}$.} equipped with a preferred decomposition in terms of $[\mathbf{M},\mathbf{M}],\otimes,$ and $\&$. Note that for each object $\mathcal{X}$ equipped with such a decomposition there is a (natural) embedding $q^{\mathcal{X}}: \mathcal{X} \rightarrow \mathcal{X}^{\otimes \rightarrow \&}$ (from the isomix structure of $\mathbf{Caus}[\mathbf{C}]$). Here $\mathcal{X}^{\otimes \rightarrow \&}$ represents the object obtained by replacing each $\otimes$ in the decomposition of $\mathcal{X}$ with a $\&$. There is furthermore an isomorphism (constructed from the $*-$autonomous structure of $\mathbf{Caus}[\mathbf{C}]$ and fact that for all $A,B$ in $[\mathbf{C}]^1$ then $A \otimes B \cong A \& B$ \footnote{see \cite{Kissinger2019AStructure} for details.}) of type $p^{\mathcal{X}}: \mathcal{X}^{\otimes \rightarrow \&} \rightarrow [\otimes_{i} \mathcal{X}_{i}^{in} , \otimes_{i} \mathcal{X}_{i}^{out} ]$ where $\mathcal{X}_{i}^{in}$ and $\mathcal{X}_{i}^{out}$ are the objects which appear on left and right hand sides of square brackets in $\mathcal{X}$ respectively. We will define by $\overline{\mathbf{M}}(A,B)$ the set of all $ \hat{\phi} \in \mathbf{C}(I,[A,B])$ such that $ \phi  \in  \mathbf{M}(A,B)$.

To define morphisms we say that $S \in \mathbf{M}^{2}( \mathcal{X} ,\mathcal{Y}) $ if and only if for every $\mathcal{E} \in ob(\mathbf{M}^2)$ and for every $\phi \in {\mathbf{C}}(I,\mathcal{X} \& \mathcal{E})  $ such that $p^{\mathcal{X} \& \mathcal{E}} \circ q^{\mathcal{X} \& \mathcal{E}} \circ (\phi) \in \overline{\mathbf{M}}(\otimes_i \mathcal{X}_i^{in} \otimes_{k} \mathcal{E}_{k}^{in} , \otimes_i \mathcal{X}_i^{out} \otimes_{k} \mathcal{E}_{k}^{out}    )$ then $p^{\mathcal{Y} \& \mathcal{E}} \circ q^{\mathcal{Y} \& \mathcal{E}} \circ (S \& I_{\mathcal{E}}) \circ \phi \in \overline{\mathbf{M}} (\otimes_j \mathcal{Y}_j^{in} \otimes_{k} \mathcal{E}_{k}^{in} , \otimes_j \mathcal{Y}_j^{out} \otimes_{k} \mathcal{E}_{k}^{out} )$. The category $\mathbf{M}^{2}$ inherits it's sequential composition and identity morphisms directly from $[\mathbf{C}]$, it then inherits both monoidal structures from $[\mathbf{C}]$. Indeed for inheritance of $\&$, consider  any pair $S_j :\overline{\mathbf{M}} (\mathcal{X}_j , \mathcal{Y}_j )$, one can confirm that $S_1 \& S_2 \in \mathcal{M}^2 (\mathcal{X}_1 \& \mathcal{X}_2 , \mathcal{Y}_1 \& \mathcal{Y}_2 )$, to check this consider that $p^{\mathcal{X}_1 \& \mathcal{X}_2 \& \mathcal{E}} \circ q^{\mathcal{X}_1 \& \mathcal{X}_2 \& \mathcal{E}} \circ (S_1 \& S_2 \& I_{\mathcal{E}}) \circ \phi = p^{\mathcal{X}_1 \& \mathcal{X}_2 \& \mathcal{E}} q^{\mathcal{X}_1 \& \mathcal{X}_2 \& \mathcal{E}} \circ (\textrm{Swap} \& I_{\mathcal{E}})  \circ  (S_2 \& I \& I_{\mathcal{E}}) \circ (\textrm{Swap} \& I_{\mathcal{E}})    \circ  ( S_1 \& I \& I_{\mathcal{E}}) \circ \phi$ and so all that is needed is to check that for any $\mathcal{X}$,$\mathcal{Y}$ then $\textrm{Swap}_{\mathcal{X}, \mathcal{Y}} \in \mathbf{M}^2  (   \mathcal{X} \& \mathcal{Y}   , \mathcal{Y} \& \mathcal{X} )$ which is immediate. Similarly for inheritance of $\otimes$, the same reasoning can be applied after noting that since $p^{(\mathcal{X}_1 \otimes \mathcal{X}_2) \& \mathcal{E}} = p^{(\mathcal{X}_1 \& \mathcal{X}_2) \& \mathcal{E}}$ and $q^{(\mathcal{X}_1 \otimes \mathcal{X}_2) \& \mathcal{E}} = q^{(\mathcal{X}_1 \otimes \mathcal{X}_2) } \& I_{\mathcal{E}}$ then by naturality of $q$ we have $p^{(\mathcal{X}_1 \otimes \mathcal{X}_2) \& \mathcal{E}} \circ q^{(\mathcal{X}_1 \otimes \mathcal{X}_2) \& \mathcal{E}} \circ ((S_1 \otimes S_2) \& I_{\mathcal{E}}) \circ \phi = p^{(\mathcal{X}_1 \& \mathcal{X}_2) \& \mathcal{E}} \circ q^{(\mathcal{X}_1 \& \mathcal{X}_2) \& \mathcal{E}}  \circ ((S_1 \& S_2) \& I_{\mathcal{E}}) \circ (q^{\mathcal{X}_1 \otimes \mathcal{X}_2 } \& I_{\mathcal{E}}) \circ \phi $,. Note that it is clear by definition that $q^{\mathcal{X}_1 \otimes \mathcal{X}_2} \in \mathbf{M}^2 (\mathcal{X}_1 \otimes \mathcal{X}_2 , \mathcal{X}_1 \& \mathcal{X}_2)$.

Finally we are ready to check that the enrichment structure of $[\mathbf{C}]$ is inherited into $\mathbf{M}^2$. We must check that $\bigcirc_{ABC} \in \mathbf{M}^2 ([A,B] \otimes [B,C] , [A,C])$ (that $\otimes_{ABA'B'} \in \mathbf{M}^2 ([A,A'] \otimes [B,B'] , [A \otimes B , A' \otimes B'])$ is immediate since $\otimes_{ABA'B'}$ is exactly $p$). First, note that for every type $([A,B] \otimes [B,C]) \& \mathcal{E}$ with $\mathcal{E} \in ob(\overline{\mathbf{M}})$ there is an embedding into $([A,B] \otimes [B,C]) \& [E,E']$ for some objects $E,E'$, this entails that every process of the former type embeds as a tripartite process satisfying \[ \forall \rho: \quad \quad  \input{figs/subcat_rho_1.tikz} , \] and so satisfies  \[   \input{figs/subcat_1.tikz} .  \]  Using the premise of this example then gives: \[    \input{figs/subcat_2.tikz} ,  \] with $L,R$ in $\mathbf{M}$. Finally, noting that the sequential composition morphism of $[\mathbf{C}]$ is constructed from the underlying compact closed structure of $\mathbf{C}$, we find that the sequential composition morphism is a morphism of $\mathbf{M}^{2}$, since it sends morphisms of $\mathbf{M}$ to morphisms of $\mathbf{M}$ in the following sense: \[ \input{figs/subcat_3.tikz}. \] An important special case of subcategories captured by this construction is given by taking $\mathbf{C} := \mathbf{CPM}[\mathbf{FHilb}]$ and taking $\mathbf{M}^1 \subseteq [\mathbf{C}]^1$ to be the category of unitary channels. In this case $\mathbf{M}^2$ is a monoidal category of unitary supermaps, such supermaps have previously been defined and are of particular interest in quantum causal modelling \cite{BarrettCyclicModels}. The monoidal and enrichment structure whilst natural, had not yet been defined to the authors knowledge. 
\end{example}

\begin{example}[Idempotent completion preserves enriched monoidal categories]
Any enriched monoidal category can be completed to include idempotents as types, meaning that decoherences in lower order theories can be inherited to construct a higher order theory with classical types. For any $\mathbf{V}$-smc $\mathbf{C}$ the idempotent completions $\mathcal{K}(\mathbf{C})$ of $\mathbf{C}$ and $\mathcal{K}(\mathbf{V})$ of $\mathbf{V}$ define a $\mathcal{K}(\mathbf{V})$-smc $\mathcal{K}(\mathbf{C})$. The idempotent completion of a category has for each object a pair $(X,x)$ of an object $X$ of $\mathbf{C}$ and an idempotent $x:X \rightarrow X$. The enriched structure is given by the functor $[(X,x),(Y,y)] = ([X,Y],[x,y])$ and the corresponding composition morphisms $\bigcirc_{xyz}$ are given by $[x,z] \circ  \bigcirc_{XYZ} \circ ([x,y] \otimes [y,z])$ and similarly for the parallel composition morphisms. All required coherences follow from coherences of the $\mathbf{V}$-smc $\mathbf{C}$. This second example, when applied to the $[\mathbf{C}]^2_{\cong}$-smc $[\mathbf{C}]^1$ with $\mathbf{C} = \mathbf{CPM}[\mathbf{FHilb}]$ of quantum supermaps over completely positive trace preserving quantum channels, produces a theory which includes classical channels in $\mathbf{C}$ as those of type $(X,dec_{X}) \rightarrow (Y,dec_{Y})$.

\end{example}

\begin{example}
For any $\mathbf{V}_{\cong}$-smc $\mathbf{C}$ one can construct the $\mathbf{V}_{\cong}^{comb}$-smc $\mathbf{C}$ which is generated by the structural morphisms $\mathbf{V}(I,[A,B]) \cup \{ \circ_{ABC} \} \cup \{  \otimes_{AA'BB'}  \} \cup \{  \texttt{coherences} \}$ of $\mathbf{V}$, meaning that all parallel composition and sequential composition supermaps are kept along with all states. This category is the category of combs of processes from $\mathbf{C}$, Indeed a comb drawn intuitively as: \[ \input{figs/comb.tikz}    \] can be formally represented as: \[ \input{figs/comb_2.tikz}    \]

Consequently all combs of processes in $\mathbf{C}$ have to exist as processes in $\mathbf{V}$ for any $\mathbf{V}_{\cong}$-smc $\mathbf{C}$.
\end{example}

\subsection{Consequences}
Some basic important consequences of monoidal enrichment are the following: There always exists a family of morphisms which represent \textit{partial insertion} and and a family functions which represent \textit{usage} of the output of a transformation. Before we begin, note that the assignment of an object $[A,B]$ in $\mathbf{V}$ to each pair $(A,B)$ of objects in $\mathbf{C}$ can be extended to an assignment on morphisms \[   [f,g]: = \input{figs/hom_functor.tikz}. \] Altogether this assignment defines a functor $\mathbf{C}^{op} \times \mathbf{C} \rightarrow \mathbf{V}$, meaning that $[f,g] \circ [f' , g'] = [f' \circ f , g \circ g']$ and $[i,i] = i$. 

\paragraph{Partial Insertion:} The partial insertion morphism $\Delta: [A,X] \otimes [Y \otimes X,Z] \rightarrow [Y \otimes A,Z]$ takes a valid sub-input of a process and inserts a pre-processing there, leaving the rest of the inputs unchanged. Formally it is defined by:
\begin{equation}
    \input{figs/define_partial_1.tikz},
\end{equation}
up to unitors and associators, where $\otimes : [Y,Y] \otimes [A,X] \rightarrow [YA,YX]$ and $\circ : [YA,YX] \otimes [YX,Z] \rightarrow [YA,Z]$. The partial insertion can be intuitively understood as representing the following picture:
\begin{equation}
    \input{figs/axiom_delta.tikz}.
\end{equation}
Crucially for $A = I$ then $\Delta$ satisfies:
\begin{equation}
    \input{figs/define_partial_2.tikz}.
\end{equation} Intuitively the above represents the equality between \begin{equation}
    \input{figs/axiom_delta_2.tikz},
\end{equation} and
\begin{equation}
    \input{figs/axiom_delta_3.tikz}.
\end{equation}
\paragraph{Usage:} The usage transformation is a particular natural transformation $\theta:{\mathbf{V}(-,[A,-])} \Longrightarrow  {\mathbf{V}([I,A] \otimes -,[I,-])} $, a family of functions $\theta_{BX}$ given by for each $S:X \rightarrow [A,B]$ taking $\theta_{BX}(S)$ to be:
    \begin{equation}
        \input{figs/theta_new.tikz}.
    \end{equation}
All of the previously stated examples have faithful usage transformations. Intuitively, $\theta$ places $S$ into one of the two holes of a sequential composition supermap: \[ \input{figs/axiom_usage.tikz}.  \] We will find that injectivity of the function $\theta$ is crucial for results about embeddings between layers of higher order theories.
\begin{definition}
A $\mathbf{V}_{\cong}$-smc $\mathbf{C}$ will be called faithful if the usage transformation \[\theta: {\mathbf{V}(-,[A,-])} \Longrightarrow  {\mathbf{V}([I,A] \otimes -,[I,-])}\] is a monomorphism in the functor category $\mathbf{Cat}(\mathbf{V}^{op} \otimes \mathbf{C},\mathbf{Set})$. 
\end{definition}
Faithful usage when present says that two higher order processes $S,T: X \rightarrow [A,B]$ should \textit{only} be distinguishable if they are distinguishable when their outputs are applied to the space of states on $A$. Stated formally faithful usage is the requirement of injectivity, that for all $I,A,B$ the composition process $\bigcirc_{IAB}$ satisfies
    \begin{equation}
        \input{figs/complete_inj.tikz}.
    \end{equation}
This has the additional consequence of entailing that the functor $[I,-]$ be faithful. We will often refer to this functor from the lower order theory to the higher order one as the \textit{raising functor} and give it the notation $\mathcal{R}:\mathbf{C} \rightarrow \mathbf{V}$.

%\begin{example}[Free category of combs]
%For any symmetric monoidal category $\mathbf{C}$ a category $\mathbf{cwpComb}$ can be constructed. TBD
%\end{example}

\subsection{Structure preserving maps}
%The standard given notion of subcategory given by the definition of an enriched monoidal category is that of a faithful monoidal $\mathbf{V}$ enriched functor, which leaves the higher order category $\mathbf{V}$ unchanged. 
The goal of this section is to develop a way of comparing theories of supermaps, the ready cooked definition of structure preserving map between enriched monoidal categories is that of a $\mathbf{V}$-enriched functor. This definition will however not be fit for our purposes since an enriched monoidal functor is one which allow comparisons of the following type: \[ \begin{tikzcd}
{(\mathbf{V},\mathbf{C})} \arrow[r] & {(\mathbf{V},\mathbf{C}')},
\end{tikzcd}    \] that is, those in which the enriching category is left untouched. Even in the basic case of the inclusion between quantum combs and quantum supermaps, functors which which fix the enriching category are unsuitable, it is after-all the enriching categories that vary in this case: \[\begin{tikzcd}
{([\mathbf{C}]^{2,comb},[\mathbf{C}]^1)} \arrow[d] \\
{([\mathbf{C}]^{2},[\mathbf{C}]^1)}           
\end{tikzcd} . \] Such an inclusion can be identified as an instance of change of base for enriched monoidal categories. Change of base alone however, will be insufficient for our purposes. Consider for instance the inclusion of the theory of unitary combs on the unitaries into the theory of quantum combs on the quantum channels  \[\begin{tikzcd}
{(\mathbf{U}^{2,comb},\mathbf{U}^{1})} \arrow[rd]   & \\
&  {([\mathbf{C}]^{2,comb},[\mathbf{C}]^1)} 
\end{tikzcd} , \] here the inclusion is neither simply an enriched functor nor simply a change of base, it is a combination of the two. 
As a more elaborate example, we will find under additional conditions of faithful usage, that for any sequence in which $\mathbf{C}^3$ enriches $\mathbf{C}^2$ which in turn enriches $\mathbf{C}^1$ then there exists a notion of embedding \[  \begin{tikzcd}
{(\mathbf{C}^2,\mathbf{C}^1)} \arrow[rd] &                               \\
                                         & {(\mathbf{C}^3,\mathbf{C}^2)}
\end{tikzcd} , \] which cannot be understood as purely an enriched functor or a change of base of enrichment. An explicit example is the embedding of the $[\mathbf{C}]^2_{\cong}$-smc $[\mathbf{C}]^1$ into the $[\mathbf{C}]^3_{\cong}$-smc $[\mathbf{C}]^2$ for any pre-causal category $\mathbf{C}$.

To address these problems and provide a suitable notion of structure preserving map, we will construct a notion of functor which allows us to vary both the lower and higher categories at the same time using combinations of (vertical) \textit{change of base} functors and (horizontal) enriched monoidal functors: \[  \begin{tikzcd}
{(\mathbf{C}^2,\mathbf{C}^1)} \arrow[rd] \arrow[d, "\texttt{Change of base}"'] &                               \\
{} \arrow[r, "\texttt{Enriched Functor}"']                                     & {(\mathbf{C}^3,\mathbf{C}^2)}
\end{tikzcd} .\]
This pair of functors along with their expected compatibility, can be viewed as an instance of what we call a pm-functor. The key components of a pm-functor are laid out explicitly here to demonstrate how they incorporate combinations of functors $\mathcal{F}^{V}:\mathbf{V} \rightarrow \mathbf{V}'$ and functors $\mathcal{F}^{C}:\mathbf{C} \rightarrow \mathbf{C}'$ along with compatibility between them. A generalised definition which takes care of coherences can be found in Appendix A, where the pm-functors are identified as arising from applying the Grothendieck construction to the change of base for $\mathbf{V}_{\cong}$ categories. 
\begin{definition}
A pm-functor from a $\mathbf{V}_{\cong}$-smc $\mathbf{C}$ to a $\mathbf{V}'_{\cong}$-smc $\mathbf{C}'$ is:
\begin{itemize}
    \item A symmetric monoidal functor $\mathcal{F}^{V}:\mathbf{V} \rightarrow \mathbf{V}'$
    \item A symmetric monoidal functor $\mathcal{F}^C: \mathbf{C} \rightarrow \mathbf{C}'$
    \item A family of morphisms $\mathcal{F}_{AB}:\mathcal{F}^V [A,B] \rightarrow [\mathcal{F}^CA,\mathcal{F}^C B]$
\end{itemize}
Which together satisfy (in functor box notation \cite{Mellies2006FunctorialDiagrams}): \[   \input{figs/smsc_monhom_r_nf_simp.tikz}  \] where for readability we have treated $\mathcal{F}^{\mathbf{C}}$ and $\mathcal{F}^{\mathbf{V}}$ as if strict, a relaxation for strong monoidal functors is given in Appendix A. \end{definition}
We will refer to a pm-morphism as \textit{fully faithful} if $\mathcal{F}^V,\mathcal{F}^C$ are full and faithful and $\mathcal{F}_{A,B}$ are isomorphisms. Indeed as a consequence we have constructed a notion of sub-enriched monoidal category, as an embedding which provides a faithful pm-morphism. 
\begin{lemma}
The pm-functors form a category $\mathbf{PM}$ with objects given by $\mathbf{V}_{\cong}$-smc's for any $\mathbf{V}$ and morphisms given by pm-functors.
\end{lemma}
\begin{proof}
This can be verified directly, with $(\mathcal{G} \circ \mathcal{F})_{AB} := \mathcal{G}_{\mathcal{F}^C A,\mathcal{F}^C B} \circ \mathcal{G}^V (\mathcal{F}_{AB})$. Alternatively, building on the results of \cite{Cruttwell2008NormedCategories}, objects and morphisms of $\mathbf{PM}$ with the above composition rule can be identified as the objects and morphisms of the Grothendieck construction for the change of base for enriched monoidal categories. This identification is made in Appendix A. 
\end{proof}

%\begin{example}[From HOQT to the Caus construction]
%For every $i,i+1$ there is a full and faithful pm-functor from the $\mathbf{HOQT}^{i+1}_{\cong}$-smc $\mathbf{HOQT}^{i}$ to the $[\mathbf{C}]^{i+1}_{\cong}$-smc $[\mathbf{C}]^{i}$. This is given by taking $\mathcal{F}^{C}A = (\mathbf{T}_{\mathbf{R}}A, \mathbf{T}_{1}A)$ and similarly for $\mathcal{F}^V$. Each is made monoidal by $(\mathbf{T}_{\mathbf{R}}AB, \mathbf{T}_{1}AB)$
%\end{example}

\begin{example}
There are pm-functors between the examples of unitary combs, general unitary supermaps, $2^{nd}$-order combs and $2^{nd}$ order transformations built from $[\mathbf{C}]$, as depicted in the preamble to this section. In each case $\mathcal{F}^C$ and $\mathcal{F}^V$ are given by inclusions and each $\mathcal{F}_{AB}$ can as a result be taken as the identity.
\end{example}
Another example with trivial $\mathcal{F}_{AB}$ is given by embedding into the Karoubi envelope.
\begin{example}
There is a pm-functor from any enriched monoidal category to its Karoubi envelope given by the embeddings $\mathbf{C} \rightarrow \mathcal{K}[\mathbf{C}]$ and $\mathbf{V} \rightarrow \mathcal{K}[\mathbf{V}]$. Explicitly we define $\mathcal{F}^{\mathbf{C}}(A) := (A,i_A)$ and simillarly $\mathcal{F}^{\mathbf{V}}(X) = (X,i_X)$ with $\mathcal{F}_{AB}$ taken to be the identity.  \end{example}
An example with non-trivial $\mathcal{F}_{AB}$ is given by considering embeddings between different layers $[\mathbf{C}]^i$ of $[\mathbf{C}]$. Note that $\mathcal{R}(-) := [I,-]$ is a lax monoidal functor with natural transformation $x_{AB}: \mathcal{R}(A) \otimes_{\mathbf{V}} \mathcal{R}(B) \rightarrow \mathcal{R}(A \otimes_{\mathbf{C}} B)$ defined by $x_{AB} := \otimes_{IAIB}$ and with morphism $y:I_{\mathbf{V}} \rightarrow \mathcal{R}(I_{\mathbf{C}})$ define by $y := \kappa(i_{\mathbf{I}_{\mathbf{C}}})$.

\begin{lemma}
Given any $\mathbf{C}^3_{\cong}$-smc $\mathbf{C}^2$ and $\mathbf{C}^2_{\cong}$-smc $\mathbf{C}^1$ there exists a morphism $\Gamma$ between them in $\mathbf{PM}$ given by given by $\Gamma := (\mathcal{R}_{2}^{3}, \mathcal{R}_{1}^{2}, \gamma)$ where $\gamma_{A,B}:\mathcal{R}_{2}^{3}([A,B]^2) \rightarrow [\mathcal{R}_{1}^2(A),\mathcal{R}_1^2 (B)]^3 $ is given by \begin{equation}
    \input{figs/smsc_monhom_2.tikz}
\end{equation}
that is, applying partial insertion in $\mathbf{C}^3$ to the composition map in $\mathbf{C}^2$.
\end{lemma}
Note that we have used the notation $[-,=]^i$ to denote the hom functor associated to the $\mathbf{C}^{i}_{\cong}$-smc $\mathbf{C}^{i-1}$. 
\begin{proof}
Given in Appendix B.
\end{proof}
Note that here $\bigcirc_{IAB}: [I,A] \otimes_{2} [A,B] \rightarrow [I,B]$ so that $\kappa(\bigcirc_{IAB}):I \rightarrow [[I,A] \otimes_{2} [A,B]]$ and we omit unitors so that $\Delta$ has type $\Delta: [I,[A,B]] \otimes_3 [[I,A] \otimes_2 [A,B], [I,B]] \rightarrow [[I,A] \otimes_{2} I,[I,B]] \cong [[I,A],[I,B]]$.

\begin{corollary}
There is a morphism in $\mathbf{PM}$ from the $[\mathbf{C}]^{i+1}_{\cong}$-smc $[\mathbf{C}]^{i}$ to the $[\mathbf{C}]^{i+2}_{\cong}$-smc $[\mathbf{C}]^{i+1}$ for every $i$. %The same story holds for $[\mathbf{C}_{\leq}]$.
\end{corollary}

\section{Self-Contained Higher Order Theories}
Current examples of higher-order quantum-like theories have a further common feature of self-containment. We mean by this the idea that all processes, no matter their higher or lower order status, 
exist in the same theory, that is, the same category. Explicit examples of theories which are self-contained in this sense are higher order quantum theory and higher order causal categories both of which are conveniently closed monoidal categories. Those familiar with classical category theory may imagine that closed monoidal structure is the appropriate mathematical formalisation of the idea of self-containment. In this section rather than taking as an axiom that currying (the key feature of closed monoidal categories) \textit{is} the appropriate formalisation of a self-contained theory of higher order transformations, we will treat this as a statement to be proven, by combining enrichment with a more relaxed formalisation of self-containment. The three operational features which we will use to derive currying and closed monoidal structure will at the intuitive level be the following: 
\begin{itemize}
    \item All processes in $\mathbf{C}$ have higher order representations \textit{in} $\mathbf{C}$ (Self-enrichment).
    \item There is an equivalence $A \cong [I,A]$ between $A$ and the higher order system $[I,A]$ representing the states of $A$ (linking).
    \item The usage transformation is faithful.
\end{itemize}
Conceptually, the first condition models the assumption that it is the same agents that can perform processes, super-processes, and so on. We model this with the notion of a $\mathbf{C}_{\cong}$-smc Category $\mathbf{C}$. The second condition is captured by the following:
\begin{definition}
A {\em Linked Monoidal Category} is a $\mathbf{C}_{\cong}$-smc $\mathbf{C}$ equipped with a monoidal natural isomorphism $\eta_A: A \rightarrow [I,A]$. 
\end{definition}
We furthermore say that a linked monoidal category $\mathbf{C}$ is faithful if it is faithful as a $\mathbf{C}_{\cong}$-smc $\mathbf{C}$. Intuitively, linked categories have enough structure to define canonical evaluation morphisms (the structural feature of closed monoidal categories) $\texttt{eval}_{AB}:A \otimes [A,B] \rightarrow B$ which apply processes to lower order objects, by using link-morphisms and sequential composition morphisms: \[  \input{figs/axiom_eval.tikz} .  \] Intuitively in the above diagram the available inputs are the bottom wire $A$ and the dotted process input of type $[A,B]$, the output wire is the top wire of type $B$. Indeed linked categories will turn out to be closed monoidal if they satisfy one additional condition - that they have faithful usage, this is proven by constructing evaluation morphisms in the above way.
\begin{lemma}
A category $\mathbf{C}$ is closed smc if and only if $\mathbf{C}$ is a faithful linked monoidal category.
\end{lemma}
\begin{proof}
A full proof is given in Appendix C, here we show for reference the formal construction of evaluations analogous to the above intuitive picture:
\begin{equation}
    \texttt{eval}  \quad : =  \quad \input{figs/main_thm_1_a.tikz}
\end{equation}
the requirement of being faithful ensures uniqueness/co-universality. 
\end{proof}
In this section we considered theories which are from the start assumed to be self contained, in the next we generalise this result to infinite towers of enriched monoidal categories, using along the way the developed notion of structure preserving map as $pm$-functor. Given that linked faithful categories are exactly closed symmetric monoidal categories they give an alternative way to view a familiar categorical structure. For instance the closed monoidal structure of $\mathbf{Set}$ can be viewed as a consequence of the fact that (i) $\mathbf{Set}$ is monoidal (ii) Trivially $\mathbf{Set}$ is monoidally enriched in $\mathbf{Set}$ (iii) The bijection $\kappa$ for enrichment provides a function $\eta : A \rightarrow [I,A]$ which is monoidal (iv) The composition function is faithful. As noted in the preliminary section $[\mathbf{C}]$ has closed monoidal structure for any $\mathbf{C}$.

%%%%%%%%%%%%%%%

\section{Towers of Higher Order Theories}
A heavy assumption used in the results of the previous section, is the idea of self-containment. In this section we will relax this assumption and show that theories resulting from the gluing together of a suitably well-behaved tower of higher order theories, are again closed monoidal. The essence of the proof will be that of the previous section, the formal tools used will be the notions of enriched monoidal category along with properties of previously constructed pm-functors (structure preserving maps) within towers of enriched categories. We begin by presenting the notion of a tower of theories over a base theory, each a theory of supermaps over the theory that precedes it. The base theory represents a given physical theory, such as quantum or classical probability theory. The second layer represents a theory of supermaps, the third layer a theory of super-supermaps, and so on:
\begin{equation}
    \input{figs/hopt_3.tikz}
\end{equation}
The ultimate goal of introducing this construction is for the specification of a unified higher order theory, into which such a sequence will embed. Mathematically, a hierarchy of higher order physical processes is represented by an {\em ascending sequence of enriched monoidal Categories}.
\begin{definition}
An {\em ascending sequence of enriched monoidal Categories} $\mathbf{C}^1 ,  \mathbf{C}^2 ,  \dots  ,  \mathbf{C}^N$ is a specification for each $i \in [\mathbf{N-1}]$ of a $\mathbf{C}^{i + 1}_{\cong}$-smc $\mathbf{C}^i$.
\end{definition}
In any such sequence, the category ${\cal C}_i$ is  ``encoded'' into the higher level ${\cal C}_{i+1}$  by the monoidal \textit{raising} functor $\mathcal{R}_{i}^{i+1}(-):= [I_{i},-]^{i+1}$. It will be convenient to define the following compact notation for the induced encoding (full faithful braided monoidal functor) from level $i$ to level $j>i$:
\[\mathcal{R}_i^j: \mathbf{C}^i \longrightarrow \mathbf{C}^j \quad \quad \mathcal{R}_i^j := \mathcal{R}_{j-1}^{j} \circ \mathcal{R}_i^{j-1}. \]

The agents inhabiting layer $\mathbf{C}_{j}$ are strictly more powerful than the inhabitants of $\mathbf{C}_{i < j}$, in the sense that each $\mathbf{C}_{i<j}$ may be embedded into $\mathbf{C}_{j}$. For any finite sequence $\mathbf{C}^{i}$ of length $N$, the final category $\mathbf{C}_{n}$ may be seen as the arena in which agents may manipulate processes from any category in the preceding sequence. In fact as observed in section $3$ the embeddings can be run in parallel and phrased as pm-morphisms.
\begin{lemma}
Let $\mathbf{C}^1 ,  \mathbf{C}^2,   \dots  $ be a {\em an ascending sequence of monoidal enriched categories}: then for every $0<i<N-1$ there exists a morphism $ \Gamma_i $ from the $\mathbf{C}^{i+1}$-smc $\mathbf{C}^i$ to the $\mathbf{C}^{i+2}_{\cong}$-smc $\mathbf{C}^{i+1}$ in $\mathbf{PM}$.
\end{lemma}
\begin{proof}
Direct consequence of lemma $1$.
\end{proof}
We will find that closed monoidal structure arises when the $\Gamma_i$ are fully faithful.
\begin{definition}
An ascending sequence of monoidal enrichments is fully faithful if each $\mathbf{C}^{i+1}_{\cong}$-smc $\mathbf{C}^i$ is faithful and each $\Gamma_i$ is fully faithful.
\end{definition}
\begin{comment}
Furthermore sequences of enriched monoidal categories admit enriched monoidal functors between their layers, again proven in the previous section:
\begin{corollary}
Let $\mathbf{C}^1   \mathbf{C}^2,   \dots    \mathbf{C}^N$ be an {\em ascending sequence of enriched monoidal r-Categories} then for every $0<i<N-1$ there exists a enriched monoidal functor \[\Gamma^{i}: \mathbf{C}^{i}   \mathbf{C}^{i+1} \rightarrow \mathbf{C}^{i+1}   \mathbf{C}^{i+2} \] an ascending sequence will be referred to as \textit{fully coherent} if every $\gamma^i$ is full and faithful
\end{corollary}
\end{comment}
Note that for each $[\mathbf{C}]$ the sequence $[\mathbf{C}]^{(1)}, [\mathbf{C}]^{2}  \dots$ is fully faithful with the inverse to each $\gamma:[I,[A,B]] \rightarrow [[I,A],[I,B]]$ given by $\eta_{[A,B]} \circ [\eta_A,\eta^{-1}_{B}]: [[I,A],[I,B]] \rightarrow [I,[A,B]]$ where $\eta_{X}: X \rightarrow [I,X]$ is the currying of the identity\footnote{In the non-strict case $\eta$ is the currying of the unitor.}. We will discover that when fully coherent sequences surjectively embed into a symmetric monoidal category, a merger for the sequence, it is guaranteed that the merger will in turn be closed monoidal.

\paragraph{Theories consisting of fully faithful sequences are closed monoidal:}
We now present operational conditions on a theory $\mathbf{C}$ in terms of an embedded ascending sequence $\mathbf{C}^{(i)}$ which will lead to closed monoidal structure for $\mathbf{C}$. Intuitively the conditions are the following \begin{itemize}
    \item $\mathbf{C}$ contains nothing more and nothing less than a fully faithful sequence $\mathbf{C}^{i}$ of enriched monoidal categories.
    \item $\mathbf{C}$ provides links between the layers of $\mathbf{C}^{i}$.
\end{itemize} Such an embedding for a generic sequence $\mathbf{C}^{(i)}$ is captured categorically by a sequence of full and faithful functors $\mathcal{F}_{i}: \mathbf{C}^{(i)} \rightarrow \mathbf{C}$. To capture that there is \textit{nothing more} in $\mathbf{C}$ we require a further condition of essential surjectivity on the union (co-product) functor \[\coprod_{i} \mathcal{F}_i : \mathbf{C}^{(i)} \rightarrow \mathbf{C} .\] 
Finally we impose the condition that there be a link between layers of the theory. The most basic notion of a linking between levels is via an isomorphism $A \cong [I,A]$. Formally this equivalence when consistent with the monoidal embeddings, is captured by the existence of a monoidal natural isomorphism 
\begin{equation} \eta_{i-1}^{i} : \mathcal{F}_{i-1}(-)  \longrightarrow \mathcal{F}_{i} \circ \mathcal{R}_{i-1}^{i} (-)  \, ,\end{equation}
In short, $\eta_{i-1}^{i}$ provides a witness for the equivalence between $A$ and $[I,A]$ inside $\mathbf{C}$. For ease of notation we will denote the inverse $(\eta^{i}_{i-1})^{-1}$ by $\eta^{i-1}_{i}$ when needed. The existence of a natural isomorphism $\eta_{i-1}^{i}$ for each $i$ is can be concisely phrased in the language of $2$-Categories, it is precisely the requirement that $\mathbf{C}$ be a $2$-Cone in the $2$-Category $\mathbf{ffSymCat}$ of 
\begin{itemize}
    \item Symmetric monoidal categories,
    \item Full and faithful symmetric lax monoidal functors,
    \item Monoidal natural transformations.
\end{itemize}
For a diagram $\mathcal{D}$ in $\mathbf{ffSymCat}$ given by a fully faithful sequence of enriched monoidal categories, and the monoidal functors $\mathcal{R}_{i}^{i+1}: \mathbf{C}^i \longrightarrow \mathbf{C}^{i + 1}$ between them, a cone over $\mathcal{D}$ is an ``apex" category $\mathbf{C}$ equipped with a family of functors $\mathcal{F}_i:\mathbf{C}^i \rightarrow \mathbf{C}$ such that each of the following triangles commutes up to a monoidal natural isomorphism $\eta_{i}^{i+1}$: 
\[\begin{tikzcd}
\mathbf{C}^i \arrow[rr, "{\mathcal{R}_{i}^{i+1}}"] \arrow[rrdd, "\mathcal{F}_{i}"'] &                                               & \mathbf{C}^{i + 1} \arrow[dd, "\mathcal{F}_{i+1}"] \\
                                                                         & {} \arrow[ru, "\eta" description, Rightarrow] &                                                   \\
                                                                         &                                               & \mathbf{C}                                      
\end{tikzcd} .  \]
The above discussion culminates in the following definition, that of a \textit{Merger}.
\begin{definition}
A Merger for a fully faithful ascending sequence of enriched monoidal categories (``merger" for short) $\mathbf{C}^{(i)}$ is a $2$-Cone $(\mathcal{F}_i: \mathbf{C}^i \longrightarrow \mathbf{C})$ over the diagram \[\begin{tikzcd}
\dots & \mathbf{C}_{i-1} \arrow[rr, "{\mathcal{R}_{i-1}^{i}}"] & & \mathbf{C}^i \arrow[rr, "{\mathcal{R}_i^{i+1}}"]  &                                               & \mathbf{C}^{i + 1} & \dots
\end{tikzcd}\] in $\mathbf{ffSymCat}$ such that
\begin{itemize}
    \item $\coprod_{i} \mathcal{F}_i$ is essentially surjective
\end{itemize}
A Merger is furthermore termed ``$N$-th order" if the sequence has length $N$.
\end{definition}

For any sequence of finite order the notion of a merger is essentially trivial, given a sequence of order $N$ one can simply construct a cone of the above type by taking $\mathbf{C} = \mathbf{C}^N$ and taking $\mathcal{F}_{i}:= \mathcal{R}_{i}^{N}$. The primary technical contribution of this manuscript is the observation that the apex of any $\infty$-Order Merger possesses a simple categorical property, it must be a closed monoidal category. 
\begin{theorem}
The apex $\mathbf{C}$ of any Merger of infinite order is a closed symmetric monoidal category.
\end{theorem}
\begin{proof}
Given in Appendix D.
\end{proof}
This result provides an operational justification for using closed symmetric monoidal categories to study higher order physics. From this position the consequences of basic physical principles in \textit{higher order} physics can be explored within a simple mathematical addition to symmetric monoidal categories. First steps in this direction of research are taken in \cite{wilsoncausality}, in which an interaction between the strength of spatial correlations, determinism, and the possibility of signalling between parties is observed.
\section{Conclusion}
Presented in this manuscript is a proposed beginning of a mathematical framework for higher order physical theories, centred around monoidal enrichment. This framework is put forward in analogy to the process theory framework for standard physics based on the notion of a symmetric monoidal category. The definitions proposed are easily iterated to define towers of theories, after-which currying in higher order theories is understood through two results: Linked faithful enriched categories are exactly closed monoidal categories, and categories into which infinite towers of higher order theories are suitably embedded, are always similarly always closed monoidal. Many open questions then follow from this point:
\begin{itemize}
    \item Do quantum supermaps, quantum combs, higher order quantum theory, and higher order causal categories, satisfy universal properties with respect to the above defined structure preserving maps, in analogy to the universal properties satisfied by CPTP maps \cite{statonDBLP:journals/corr/abs-1901-10117}?
    \item It is not clear that \textit{every} enrichment gives a valid example of a theory of supermaps, so what other properties should be expected? In particular, how could the key additional feature of local-applicability \cite{Chiribella2008TransformingSupermaps, Wilson2022QuantumLocality, Wilson2022FreeDimension} be combined with enriched monoidal structure? 
    \item Recent constructions in the literature suggest that in very general settings it may be more natural to think of the theory $\mathbf{V}$ as being multicategorical \cite{Wilson2022QuantumLocality}, polycategorical \cite{Wilson2022FreeDimension, Hefford_2023_comb}, or promonoidal \cite{earnshaw2023produoidal}, rather than plainly monoidal. Leading to a natural question of whether analogous tower theorems can be constructed for embedding into variants of closed structures in these more relaxed settings. 
    \item How can we include supermaps on constrained spaces such as one-way signalling channels and routed quantum supermaps \cite{Vanrietvelde_2021, vanderrouted, wilson_constraints}?
    \item What does the view of higher order physics as enriched monoidal structure have to say about resource theories of higher order processes \cite{takagiresourcechannels, Kristjansson2020ResourceCommunication, gour2020dynamical, Gonda2023monotonesingeneral, Coecke2014AResources, FRITZ_2015_resource}?
\end{itemize}
This paper aims to lay a basic starting point from which the above questions can be formulated and answered.

%%%%%%%%%%%%%%%%%

%We conclude that the assumptions of a linked internal enriched monoidal category, 
%\begin{itemize}
%    \item There is a static version of any process, whether higher or lower order, which can be manipulated by sequential or parallel composition. ($\mathbf{C}$ is a enriched monoidal category over $\mathbf{C}$)
%    \item There is a way for agents to move back and forth between layers, (A natural isomorphism $[I,A] \cong A$) 
%\end{itemize}
%give an operational characterisation of the above stated class of CSMCs.

\begin{acknowledgments}
MW would like to thank A Vanrietvelde, J Hefford, P Selinger, B Coecke, and G Boisseau for useful conversations. This work is supported by the Hong Kong Research Grant Council through grant 17300918 and though the Senior Research Fellowship Scheme SRFS2021-
7S02, by the Croucher Foundation, by the John Templeton Foundation through grant 61466, The Quantum Information Structure of Spacetime (qiss.fr). Research at the Perimeter Institute is supported by
the Government of Canada through the Department of Innovation, Science and Economic Development
Canada and by the Province of Ontario through the Ministry of Research, Innovation and Science. The
opinions expressed in this publication are those of the authors and do not necessarily reflect the views of
the John Templeton Foundation. MW gratefully acknowledges support by University College London
and the EPSRC Doctoral Training Centre for Delivering Quantum Technologies through grant number [EP/L015242/1].
\end{acknowledgments}

\newpage

\bibliographystyle{plainnat}
\bibliography{references_link_format}

\appendix

\section{Enriched Monoidal Categories}
A key notion for the description of supermaps will be that of an enriched category \cite{Kelly2005}, in this section we will find a well behaved category in which $\mathbf{V}_{\cong}$-smc's live and use this category to see pm-functors as morphisms arising from the Grothendieck construction.
Note that from any object $B$ of a monoidal category $\mathbf{V}$ and functor $[-,-]: \mathbf{C}^{op} \times \mathbf{C} \rightarrow \mathbf{V}$ a new functor $[-,B] \otimes [B,-]: \mathbf{C}^{op} \times \mathbf{C} \rightarrow \mathbf{V}$ can be constructed.
\begin{definition}
Let $\mathbf{C}$ be a category and $\mathbf{V}$ be a symmetric monoidal category, a $\mathbf{V}_{\kappa}$ enriched category $\mathbf{C}$ is a specification of
\begin{itemize}
    \item A hom-functor $[A,B]: \mathbf{C} \times \mathbf{C} \rightarrow \mathbf{V}$
    \item A natural transformation  ${\circ}_{ABC} : [A,B] \otimes [B,C] \rightarrow [A,C]$ (with a minor abuse of notation in the use of the symbol $\circ$)
    \item A natural transformation $\kappa: \mathbf{C}(A,B) \rightarrow \mathbf{V}(I,[A,B]])$ 
    \item Such that the following unitality and associativity conditions for the natural transformation $\circ_{ABC}$ hold:
        \begin{equation}
    \input{figs/smsc_asso1.tikz}
\end{equation}
\end{itemize}
\end{definition}
A $\mathbf{V}_{\cong}$ enriched category $\mathbf{C}$ is defined as a particular kind of $\mathbf{V}_{\kappa}$ enriched category. One might reasonably wonder why we go to the effort of defining $\mathbf{V}_{\kappa}$ enriched categories, this is to cope with a subtlety in which we would like to define a $\mathbf{V}_{\cong}$-smc $\mathbf{C}$ as a pseudomonoid, it turns out that such pseudomonoids are most naturally defined with respect to the $2$-monoidal structure of the category of $\mathbf{V}_{\kappa}$ categories into which the category of $\mathbf{V}_{\cong}$ categories embeds. 
\begin{definition}
A $\mathbf{V}_{\cong}$ enriched category $\mathbf{C}$ is a $\mathbf{V}_{\kappa}$ enriched category in which $\kappa$ is a natural isomorphism. 
\end{definition}
From the definition of $\mathbf{V}_{\kappa}$-categories one can define $\mathbf{V}_{\kappa}$-functors and $\mathbf{V}_{\kappa}$-natural transformations.
\begin{definition}
A $\mathbf{V}_{\kappa}$-functor between $\mathbf{V}_{\kappa}$ categories is a functor $\mathcal{F}^C:\mathbf{C} \rightarrow \mathbf{C}'$ and a natural transformation $\mathcal{F}_{A,B}:[A,B] \rightarrow [\mathcal{F}^C A, \mathcal{F}^C B]$ satisfying: \[  \tikzfigscale{0.7}{figs/vfunct_1}.  \] A $\mathbf{V}_{\cong}$ functor is simply a $\mathbf{V}_{\kappa}$ functor between $\mathbf{V}_{\cong}$ categories. 
\end{definition}

\begin{definition}
A $\mathbf{V}_{\kappa}$-natural transformation between $\mathbf{V}_{\kappa}$ functors $\eta : \mathcal{F} \Rightarrow \mathcal{G}$ is a natural transformation $\eta^C: \mathcal{F}^C \Rightarrow \mathcal{G}^{C}$ satisfying: \[  \tikzfigscale{0.7}{figs/vnat_1} . \] A $\mathbf{V}_{\cong}$ natural transformation is a $\mathbf{V}_{\kappa}$ natural transformation between $\mathbf{V}_{\cong}$ functors. 
\end{definition}

All together $\mathbf{V}_{\kappa}\mathbf{Cat}$ defines a $2$-category, with composition of $\mathbf{V}_{\kappa}$-functors $(\mathcal{F} \circ \mathcal{G})_{A,B} := \mathcal{F}_{A,B} \circ \mathcal{G}_{A,B}$. The $2$-category $\mathbf{V}_{\kappa}\mathbf{Cat}$ is furthermore monoidal, with the monoidal product (which we will refer to as $\boxtimes_{\mathbf{V}}$) of a $\mathbf{V}_{\kappa}$ category $\mathbf{C}$ with a $\mathbf{V}_{\kappa}$ category $\mathbf{D}$ defined as the $\mathbf{V}_{\kappa}$ category $\mathbf{C} \times \mathbf{D}$ with:
\begin{itemize}
\item hom functor given by $[-,=]_{\mathbf{C}} \otimes_{\mathbf{V}} [-,=]_{\mathbf{D}}$.
\item natural transformation $\circ^{\mathbf{C} \otimes_{\mathbf{V}} \mathbf{D}} := \circ^{\mathbf{C}} \otimes_{\mathbf{V}} \circ^{\mathbf{D}}$ (up to a swap).
\item natural transformation $\kappa(f,g) :=(\kappa^{\mathbf{C}}(f) \otimes \kappa^{\mathbf{D}}(g)) \circ \lambda_{I}$.
\end{itemize}
The unit object of $\mathbf{V}_{\kappa} \mathbf{Cat}$ is given by the $(\mathbf{V}_{\kappa}$ category $\{ \bullet \}$ (the singleton category) with enrichment given by $[\bullet,\bullet] := I_{\mathbf{V}}$. All together the monoidal and $2$ categorical structures of $\mathbf{V}_{\kappa} \mathbf{Cat}$ are compatible in the sense that $\mathbf{V}_{\kappa} \mathbf{Cat}$ forms a symmetric monoidal $2$ category. For the following let $\mathbf{SymMonCat}$ be the category of symmetric monoidal categories and symmetric monoidal functors, and let $\mathbf{SymMon2Cat}$ be the category of symmetric monoidal $2$ categories and the symmetric (lax) monoidal $2$ functors between them. 
%\begin{definition}
%A $\mathbf{V}_{\cong}$-monoidal category $\mathbf{C}$ is a pseudomonoid in $\mathbf{V}_{\kappa} \mathbf{Cat}$. Explicitly this gives the category $\mathbf{C}$ monoidal structure and the category $\mathbf{V}$ an associative family of maps $\otimes_{AA'BB'}:[A,A'] \otimes [B,B'] \rightarrow [A \otimes_{\mathbf{C}} B, A' \otimes_{C} B']$ in $\mathbf{V}$ which implement the monoidal structure of $\mathbf{C}$.
%\end{definition}
%A $2$-category can immediately be defined, the category $\mathbf{VMonCat}_{\cong}$ given by the free construction of a $2$-category of pseudomonoids from any monoidal $2$-category. 

%For every lax monoidal functor $\mathcal{F}:\mathbf{V} \rightarrow \mathbf{V}'$ one can construct a $2$-functor $\mathcal{M}^{\kappa}(\mathcal{F}): \mathbf{V}_{\kappa}\mathbf{Cat} \rightarrow \mathbf{V'}_{\kappa}\mathbf{Cat}$. 
%This functor is given by sending each $\mathbf{V}_{\kappa}$ category $\mathbf{C}$ to the $(\mathbf{V}_{id}' $ category $ \mathbf{U}$ where $\mathbf{U}$ is the underlying category of the standard change of base of enriched categories.
\begin{lemma}
There is a functor $\mathcal{M}^{\kappa}: \mathbf{SymMonCat} \rightarrow \mathbf{SymMon2Cat}$ which sends every symmetric lax monoidal functor $\mathcal{F}: \mathbf{V} \rightarrow \mathbf{V}'$ to a symmetric lax monoidal $2$-functor $\mathcal{M}^{\kappa}(\mathcal{F}): \mathbf{V}_{\kappa}\mathbf{Cat} \rightarrow \mathbf{V}_{\kappa}' \mathbf{Cat} $.
\end{lemma}
\begin{proof}
%We define $\mathcal{M}^{\kappa}$ by $\mathcal{Y}_{\mathbf{V}'} \circ  \mathcal{M}_{*} \circ \mathcal{X}_{\mathbf{V}}$, where $\mathcal{X}_{\mathbf{V}}:\mathbf{V}_{\kappa} \mathbf{Cat} \rightarrow \mathbf{VCat}$ is the monoidal $2$-functor which replaces each category $\mathbf{C}$ with the underlying category $\mathbf{U}$ from enrichment and $\mathcal{Y}_{\mathbf{V}}:\mathbf{VCat} \rightarrow \mathbf{V}_{\kappa}\mathbf{Cat}$ is the monoidal $2$-functor which views each $\mathbf{V}$ monoidal category $\mathbf{U}$ as a $\mathbf{V}_{\kappa}$ monoidal category $\mathbf{U}$ using the identity functions to construct a natural transformation.
This is a minor extension of the result of \cite{Cruttwell2008NormedCategories} in which it is shown that the change of base functor for standard enriched categories $\mathcal{M}_{*}(\mathcal{F}):\mathbf{VCat} \rightarrow \mathbf{V'Cat}$ is a monoidal $2$-functor. 

Let $\mathcal{F}$ be symmetric lax monoidal, with natural transformation $x_{A,A'}: \mathcal{F}(A) \otimes \mathcal{F}(A') \rightarrow \mathcal{F}(A \otimes A')$ and morphism $y: I_{V'} \rightarrow \mathcal{F}(I_V)$. We define $\mathcal{M}^{\kappa}(\mathcal{F})$ to send each $\mathbf{V}_{\kappa}$ category $\mathbf{C}$ to the new $\mathbf{V}'_{\kappa}$ category $\mathbf{C}$ with hom functor $\mathcal{F}[-,=]$ along with sequential composition morphisms given by $\mathcal{F}(\circ_{ABC}) \circ x_{[A,B][B,C]}$ and natural transformation $\kappa^{\mathcal{F}}: \mathbf{C}(A,B) \rightarrow \mathbf{V}'(I_{\mathbf{V}'} , \mathcal{F}^{\mathbf{V}}[A,B]) $ given by $\mathcal{F} (\kappa(-)) \circ y$. 

On morphisms define $ \mathcal{M}^{\kappa}(\mathcal{F})(\mathcal{G})^{\mathbf{C}} := \mathcal{G}^{\mathbf{C}}$ and define $ \mathcal{M}^{\kappa}(\mathcal{F})(\mathcal{G})_{AB} := \mathcal{F}(\mathcal{G}_{AB})$, functorality of $\mathcal{M}^{\kappa}(\mathcal{F})(-)$ is then directly inherited from functorality of $\mathcal{F}$. On $2$-morphisms define $ \mathcal{M}^{\kappa}(\mathcal{F})(\eta) := \eta$, functorality is then immediate and the $\mathbf{V}'_{\kappa}$ naturality of $\eta$ follows from $\mathbf{V}_{\kappa}$ naturality of $\eta$ and the functorality of $\mathcal{F}$. For $\mathcal{M}^{\kappa}(\mathcal{F})(-)$ to be lax monoidal construct the required natural transformation $\mathcal{X} : \mathcal{M}^{\kappa}(\mathcal{F})(-) \boxtimes_V  \mathcal{M}^{\kappa}(\mathcal{F})(=) \Rightarrow \mathcal{M}(\mathcal{F})(- \boxtimes_{V} =)$ must for every pair of $\mathbf{V}_{\kappa}$ categories $\mathbf{C},\mathbf{D}$ give a $\mathbf{V}_{\kappa}'$ functor $\mathcal{X}^{\mathbf{C}, \mathbf{D}}$, this $\mathbf{V}_{\kappa}'$ functor is defined component-wise by $\mathcal{X}^{\mathbf{C}, \mathbf{D}}_{(A,B),(A',B')} : = x_{[A,A'][B,B']}$. On the underlying categories $\mathbf{C} \times \mathbf{D}$ then  $\mathcal{X}^{\mathbf{C}, \mathbf{D}}$ is defined to be the identity, with the $\mathbf{V}_{\kappa}$ compatibility law inherited from the unitality law for $\mathcal{F}$. That $\mathcal{X}$ is a $\mathbf{V}'_{\kappa}$-functor follows from naturality of $x$ and the associativity law for $\mathcal{X}$ is inherited from associativity of $x$. The required morphism $\mathcal{Y}: I_{\mathbf{V}'_{\kappa} \mathbf{Cat}} \rightarrow \mathcal{M}^{\kappa}(\mathcal{F})(I_{\mathbf{V}_{\kappa} \mathbf{Cat}})$ is constructed by $\mathcal{Y}_{\bullet, \bullet} :=  y$. The symmetry for $\mathcal{M}^{\kappa}(\mathcal{F})$ is inherited from the symmetry of $\mathcal{F}$.

Finally, note that $\mathcal{M}^{\kappa}(\mathcal{G} \circ \mathcal{F}) = \mathcal{M}^{\kappa}(\mathcal{G})  \circ \mathcal{M}^{\kappa}(\mathcal{F})$. Indeed, when applied to some $\mathbf{V}_{\kappa}$-smc $\mathbf{C}$ they return the same $\mathbf{V}^{''}_{\kappa}$-smc $\mathbf{C}$ as is constructed by applying first $\mathcal{M}^{\kappa}(\mathcal{F})$ and then $\mathcal{M}^{\kappa}(\mathcal{F})$. Both return a final hom functor $\mathcal{G} \mathcal{F} [-,=]$, both return the same composition natural transformation since $x_{MN}^{\mathcal{G} \mathcal{F}} = \mathcal{G}(x_{MN}^{\mathcal{F}}) \circ x_{\mathcal{F} M \mathcal{F}N}^{\mathcal{G}}$, and so $\mathcal{G}\mathcal{F}(\circ_{ABC}) \circ x_{[A,B][B,C]}^{\mathcal{G} \mathcal{F}} = \mathcal{G}(\mathcal{F}(\circ_{ABC}) \circ x_{[A,B][B,C]}^{\mathcal{F}}) \circ x_{\mathcal{F}[A,B]\mathcal{F}[B,C]}^{\mathcal{G}}$ by functorality of $\mathcal{G}$. Finally, both return the same natural transformation $\kappa^{''}$ since $y^{\mathcal{G} \mathcal{F}} = \mathcal{G}(y^{\mathcal{F}}) \circ y^{\mathcal{G}}$ and so $\mathcal{G}\mathcal{F}(\kappa(f)) \circ y^{\mathcal{G} \mathcal{F}} = \mathcal{G}(\mathcal{F}(\kappa(f)) \circ y^{\mathcal{F}}) \circ y^{\mathcal{G}}$ by functorality of $\mathcal{G}$. Finally, on morphisms $\mathcal{M}^{\kappa}(\mathcal{G} \circ \mathcal{F})(\mathcal{H})_{AB} = \mathcal{}G \mathcal{F} (\mathcal{H}_{AB}) = \mathcal{M}^{\kappa}(\mathcal{G})(\mathcal{M}^{\kappa}(\mathcal{F})(H))_{AB}$. 
%When $\mathbf{V}$ is a symmetric monoidal category and $\mathcal{F}:\mathbf{V} \rightarrow \mathbf{V}'$ is a symmetric monoidal functor the change of base functor $\mathcal{M}:\mathbf{SymMonCat} \rightarrow \mathbf{SymMon2Cat}$ assigns to each $\mathcal{F}:\mathbf{V} \rightarrow \mathbf{V}'$ a symmetric monoidal $2$-functor $\mathcal{M}^{*}(\mathcal{F}):\mathbf{VCat} \rightarrow \mathbf{V'Cat}$ indeed the braiding coherence required for $\mathcal{M}^{*}(\mathcal{F})$ to be a symmetric monoidal $2$-functor is directly inherited point-wise from the coherence for $\mathcal{F}$. Each of $\mathcal{X}_{\mathbf{V}}$ and $\mathcal{Y}_{\mathbf{V}}$ are easily verified to be symmetric, as a result the concatenation $\mathcal{M}^{\kappa}(\mathcal{F}:\mathbf{V} \rightarrow \mathbf{V}') : = \mathcal{Y}_{\mathbf{V}'} \circ  \mathcal{M}_{*}(\mathcal{F} )\circ \mathcal{X}_{\mathbf{V}}$ is therefore also symmetric. 
\end{proof}
One can construct freely construct an entire $2$-category of symmetric pseudomonoids \cite{paddy_coalgebroids, lucy_relativesymmetric} $\mathbf{SymPsMon}[\mathbf{VCat}_{\kappa}]$, this will provide a higher-level way to see that pm-functors as written explicitly without coherences in the main text can be defined and composed.
\begin{lemma}
There is a functor $\mathcal{M}^{\kappa}_{smc}:\mathbf{SymMonCat} \rightarrow \mathbf{2Cat}$ which assigns to each $\mathcal{F}:\mathbf{V} \rightarrow \mathbf{V}'$ a $2$-functor $\mathcal{M}^{\kappa}_{smc}(\mathcal{F}):\mathbf{SymPsMon}[\mathbf{V}_{\kappa} \mathbf{Cat}] \rightarrow \mathbf{SymPsMon}[\mathbf{V'}_{\kappa}\mathbf{Cat}]$.
\end{lemma}
\begin{proof}
The $2$-functor $\mathbf{SymPsMon}:\mathbf{SymMon2Cat} \rightarrow  \mathbf{2Cat}$ constructed in \cite{lucy_relativesymmetric} can be applied to return $\mathcal{M}^{\kappa}_{smc}(\mathcal{F}) := \mathbf{SymPsMon}[\mathcal{M}^{\kappa}(\mathcal{F})]:\mathbf{SymPsMon}[\mathbf{V}_{\kappa}\mathbf{Cat}] \rightarrow \mathbf{SymPsMon}[\mathbf{V'}_{\kappa} \mathbf{Cat}]$.
\end{proof}

We can now extend the definition of a $\mathbf{V}_{\kappa}$ enriched category to define a $\mathbf{V}_{\kappa}$ symmetric monoidal category ($\mathbf{V}_{\kappa}$-smc for short) as exactly a symmetric pseudomonoid \cite{DAY199799, paddy_coalgebroids} in $\mathbf{V}_{\kappa} \mathbf{Cat}$. We also make use of the monoidal structure of $\mathbf{V}_{\kappa} \mathbf{Cat}$ to define $\mathbf{V}_{\cong}$-smc's. \begin{definition}
A $\mathbf{V}_{\cong}$-smc $\mathbf{C}$ is a symmetric pseudomonoid \cite{paddy_coalgebroids} on a $\mathbf{V}_{\cong}$ category $\mathbf{C}$ in $\mathbf{V}_{\kappa} \mathbf{Cat}$.
\end{definition}
This definition of a $\mathbf{V}_{\cong}$-smc $\mathbf{C}$ as a symmetric pseudomonoid recovers the definition of a $\mathbf{V}_{\cong}$-smc $\mathbf{C}$ as given in the main text except in that it also allows for non-strictness of monoidal categories and the functors between them. To be precise, the definition given in the main text is in-fact what is recovered when in the special case that the symmetric pseudomonoid is in fact a monoid. The advantage of the rephrasing in terms of pseudomonoids, is that it neatly packages coherences for us. 

The Grothendieck construction sends any functor of type $\mathcal{M}:\mathbf{D} \rightarrow \mathbf{Cat}$ to a category $\mathcal{G}(\mathcal{M})$ with objects given by pairs $(d,c)$ where $d$ is an object of $\mathbf{D}$ and $c$ is an object of $\mathbf{M}(d)$. Morphisms $(d_1,c_1) \rightarrow (d_2,c_2)$ are given by pairs $(f,g)$ of morphisms $f:d_1 \rightarrow d_2$ and $g:\mathcal{M}(f)(c_1) \rightarrow c_2$. 

\begin{theorem}
The restriction of $\mathcal{G}(\mathcal{M}^{\kappa}_{smc})$ to objects $(\mathbf{V}, x)$ in which $x$ is a pseudomonoid on a $\mathbf{V}_{\cong}$ category, has objects given by $\mathbf{V}_{\cong}$-smc's and morphisms given by pm-functors. Concretely, a morphism in this category is given by
 \begin{itemize}
    \item A symmetric monoidal functor $\mathcal{F}^{V}:\mathbf{V} \rightarrow \mathbf{V}'$
    \item A symmetric monoidal functor $\mathcal{F}^C: \mathbf{C} \rightarrow \mathbf{C}'$
    \item A family of morphisms $\mathcal{F}_{AB}:\mathcal{F}^V [A,B] \rightarrow [\mathcal{F}^CA,\mathcal{F}^C B]$
\end{itemize}
Which together satisfy \[   \input{figs/smsc_monhom_r_nf.tikz},  \] where $x^{\mathbf{C}}:\mathcal{F}^{\mathbf{C}}(-) \otimes \mathcal{F}^{\mathbf{C}}(=) \rightarrow \mathcal{F}^{\mathbf{C}}(- \otimes =)$ and $y: I_{\mathbf{C'}} \rightarrow \mathcal{F}^{\mathbf{C}}(I_{\mathbf{C}})$ are the structural isomorphisms which make $\mathcal{F}^{\mathbf{C}}$ strong monoidal (and similarly for $x^{\mathbf{V}} , y^{\mathbf{V}} , \mathcal{F}^{\mathbf{V}}$). Therefore, $\mathbf{PM}$ is a category. 
\end{theorem}
Note that when both $\mathcal{F}^{\mathbf{V}}$ and $\mathcal{F}^{\mathbf{C}}$ are strict this exactly recovers the simplified definition of the main text. 
\begin{proof}
Note that $\mathcal{G}(\mathcal{M}^{\kappa}_{smc})$ has for each object a pair of an smc $\mathbf{V}$ and a $\mathbf{V}_{\kappa}$-smc $\mathbf{C}$ and for each morphism a pair of a symmetric monoidal functor $\mathcal{F}^V : \mathbf{V} \rightarrow \mathbf{V}'$ and a symmetric monoidal morphism\footnote{For the explicit definition of symmetric monoidal morphism between symmetric psuedomonoids see  \cite{paddy_coalgebroids}.} from $\mathcal{M}(\mathcal{F}^V)(\mathbf{V},\mathbf{C})$ to $ (\mathbf{V}',\mathbf{C}')$. Explicitly, a symmetric monoidal morphism is in this case a $\mathbf{V}_{\kappa}$ functor from $\mathcal{M}(\mathcal{F}^V)(\mathbf{V},\mathbf{C})$ to $ (\mathbf{V}',\mathbf{C}')$ satisfying additional (monoidal) properties. The bare $\mathbf{V}_{\kappa}$ functor part gives a symmetric lax monoidal functor $\mathcal{F}^{C}: \mathbf{C} \rightarrow \mathbf{C}'$ and a family of morphisms $\mathcal{F}_{A,B}:\mathcal{F}[A,A'] \rightarrow [\mathcal{F}^C A, \mathcal{F}^C A']$. Note that the structure of $\mathcal{M}(\mathcal{F}^V)(\mathbf{V},\mathbf{C})$ is outlined in lemma $5$, it has objects given by those of the $\mathbf{V}_{\kappa}$-smc $\mathbf{C}$, and has underlying category $\mathbf{C}$ with natural transformation given by  $\mathcal{F}^V (\kappa(-)) \circ y$ and sequential composition maps given by $\mathcal{F}^V(\circ_{ABC}) \circ x_{[A,B][B,C]}$.  The first and second laws for $\mathbf{V}_{\kappa}$ functors then recover the first and second laws for pm-functors. 
Now the additional $2$-morphisms in $\mathbf{V}_{\kappa} \mathbf{Cat}$ which give this $\mathbf{V}_{\kappa}$-fucntor the structure of a symmetric monoidal morphism gives $\mathcal{F}^{\mathbf{C}}$ the structure of a symmetric monoidal morphism in $\mathbf{SymPsMon}[\mathbf{Cat}]$ (and so the structure of a symmetric strong monoidal functor), but furthermore $2$-morphisms in $\mathbf{V}_{\kappa}$ are required to satisfy the additional compatibility law of definition $10$. This law requires that $[i,x^{\mathbf{C}}_{A'B'}] \circ  \otimes_{AA'BB'}' \circ (\mathcal{F}_{A,A'} \otimes_{V} \mathcal{F}_{B,B'}) = [x^{\mathbf{C}}_{AB}, I] \circ \mathcal{F}_{A \otimes B,A' \otimes B'} \circ \mathcal{F}^{V}(\otimes_{AA'BB'}) \circ x_{[A,A'][B,B']}$ the final law for pm-functors. 
\end{proof}

%Unpacking the definition of $\mathbf{PM}$ recovers the definition of the main text, expect that again coherences are taken care of by defining $\mathbf{PM}$ in terms of pseudomonoids and their homomorphisms. As a result of identifying $\mathbf{V}_{\cong}$-smc's and the pm-functors between them as the objects and morphisms arising from restricting the outcome of the Grothendieck construction to some subclass of objects, it follows immediately that pm-functors can be composed (assiciatively and unitaly) to form the category $\mathbf{PM}$.

%%%%%%%%%%%%%%%%%%%%%%%%%%%%%%%%%%%%%%%%%%%%%%%%%%%%%%%%%%%%%%%%%%%%%%%%%%%

%%%%%%%%%%%%%%%%%%%%%%%%%%%%%%%%%%%%%%%%%%%%%%%%%%%%%%%%%%%%%%

%%%%%%%%%%%%%%%%%%%%%%%%%%%%%%%%%%%%%%%%%%%%%%%%%%%%%%%%%%%%%%%%%%%5
\section{PM morphisms between layers of the Grothendieck Construction}

\begin{lemma}
Given any $\mathbf{C}^3_{\cong}$ monoidal category $\mathbf{C}^2$ and $\mathbf{C}^2_{\cong}$ monoidal category $\mathbf{C}^1$ there exists a morphism $\Gamma$ between them in $\mathbf{PM}$ given by given by $\Gamma := (\mathcal{R}_{2}^{3}, \mathcal{R}_{1}^{2}, \gamma^i)$ where $\gamma^{i}_{A,B}:\mathcal{R}_{2}^{3}([A,B]^2) \rightarrow [\mathcal{R}_{1}^2(A),\mathcal{R}_1^2 (B)]^3 $ is given by \begin{equation}
    \input{figs/smsc_monhom_2.tikz}
\end{equation}
\end{lemma}

\begin{proof}
Including the isomorphisms associated to the monoidal functors in the definition of a pm-morphisms as in Appendix A a pm-functor must satisfy: \[   \input{figs/smsc_monhom_r_nf.tikz}.  \]
For the first state preservation condition: 
\begin{equation}
    \input{figs/state_proof_0.tikz} \quad = \quad      \input{figs/state_proof_1.tikz} \quad = \quad      \input{figs/state_proof_2.tikz} \quad = \quad      \input{figs/state_proof_3.tikz} \quad = \quad  \begin{tikzpicture}[scale=0.6, {every node/.style}={scale=0.6}]
	\begin{pgfonlayer}{nodelayer}
		\node [style=box] (54) at (0, -0.5) {$\kappa_{2}(\mathcal{F}^C (f))$};
		\node [style=none] (55) at (0, 0.25) {};
	\end{pgfonlayer}
	\begin{pgfonlayer}{edgelayer}
		\draw (55.center) to (54);
	\end{pgfonlayer}
\end{tikzpicture}
 
\end{equation}
For the second condition for pm-functors: preservation of sequential composition. 
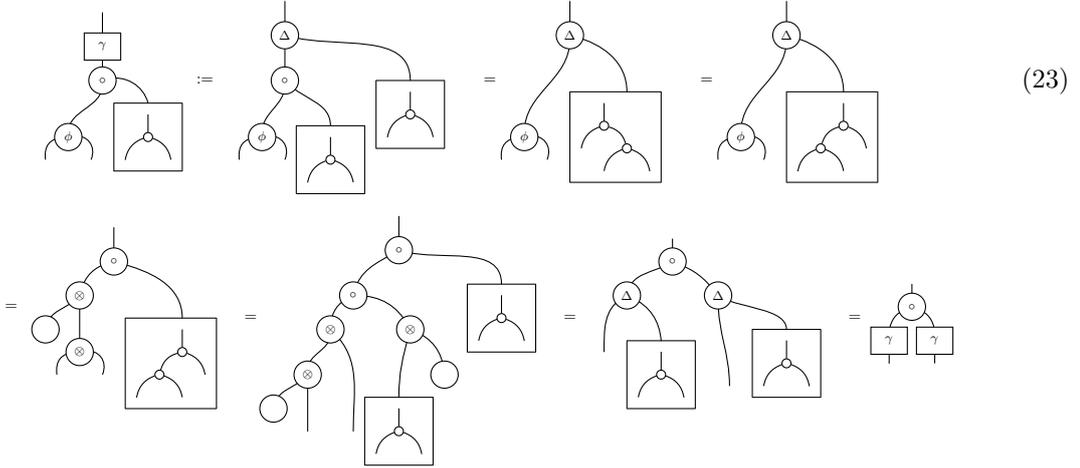
\begin{equation}
    \input{figs/prove_enrich_2a0.tikz} \quad = \quad  \input{figs/prove_enrich_2a1.tikz} \quad = \quad   \input{figs/prove_enrich_2a2.tikz} \quad = \quad   \input{figs/prove_enrich_2a3.tikz} \quad = \quad   \input{figs/prove_enrich_2a4.tikz} 
\end{equation}
\begin{equation}
 \quad = \quad   \input{figs/prove_enrich_2a5.tikz} \quad = \quad   \input{figs/prove_enrich_2a6.tikz} \quad = \quad  \input{figs/prove_enrich_2a6b.tikz}  
\end{equation}
\begin{equation}
\quad = \quad    \input{figs/prove_enrich_2a6c.tikz}  \quad = \quad   \input{figs/prove_enrich_2a6d.tikz} \quad = \quad   \input{figs/prove_enrich_2a6e.tikz} \end{equation}
\begin{equation}
\quad = \quad   \input{figs/prove_enrich_2a7.tikz} \quad = \quad   \input{figs/prove_enrich_2a8.tikz}  \quad = \quad   \input{figs/prove_enrich_2a9.tikz} 
\end{equation}
Then the third condition, for preservation of parallel composition:
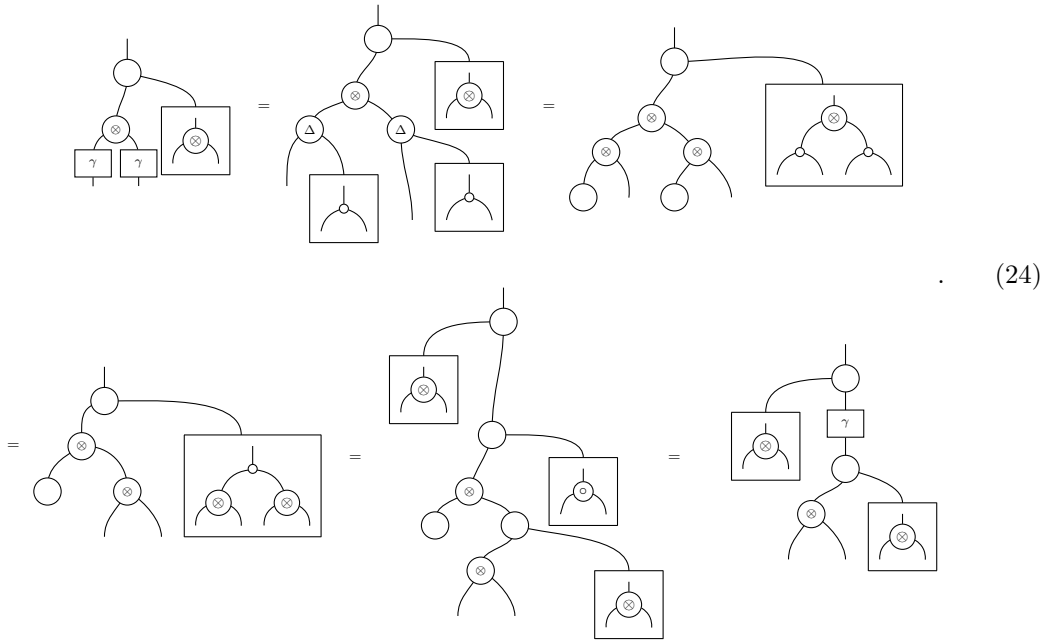
\begin{equation}
    \input{figs/prove_enrich_par_end.tikz}  \quad = \quad \input{figs/prove_enrich_par_1.tikz}  \quad = \quad \input{figs/prove_enrich_par_2b.tikz} 
\end{equation}
\begin{equation}
   \quad = \quad \input{figs/prove_enrich_par_2c.tikz} \quad = \quad    \input{figs/prove_enrich_par_3.tikz} 
   \end{equation}
   \begin{equation}
   \quad = \quad \input{figs/prove_enrich_par_3b.tikz} \quad = \quad    \input{figs/prove_enrich_par_3c.tikz}
    \end{equation}
    
    \begin{equation}
   \quad = \quad \input{figs/prove_enrich_par_3d.tikz} \quad = \quad    \input{figs/prove_enrich_par_3e.tikz}
    \end{equation}
    
        \begin{equation}
   \quad = \quad \input{figs/prove_enrich_par_3f.tikz} \quad = \quad    \input{figs/prove_enrich_par_3g.tikz}
    \end{equation}
    
            \begin{equation}
   \quad = \quad \input{figs/prove_enrich_par_3h.tikz} \quad = \quad    \input{figs/prove_enrich_par_3i.tikz}
    \end{equation}
    
                \begin{equation}
   \quad = \quad \input{figs/prove_enrich_par_3j.tikz} \quad = \quad    \input{figs/prove_enrich_par_3k.tikz}
    \end{equation}
    
                \begin{equation}
   \quad = \quad \input{figs/prove_enrich_par_3l.tikz} \quad = \quad    \input{figs/prove_enrich_par_3m.tikz}
    \end{equation}
    
                \begin{equation}
   \quad = \quad \input{figs/prove_enrich_par_3n.tikz} \quad = \quad    \input{figs/prove_enrich_par_3o.tikz}
    \end{equation}
    
                \begin{equation}
   \quad = \quad \input{figs/prove_enrich_par_3p.tikz} \quad = \quad    \input{figs/prove_enrich_par_3q.tikz}
    \end{equation}
    
\begin{equation}
  \quad = \quad   \input{figs/prove_enrich_par_3r.tikz}  \quad = \quad \input{figs/prove_enrich_par_begin.tikz}\end{equation}

\end{proof}

\begin{comment}
\begin{equation}
    \input{figs/prove_enrich_1a.tikz}
\end{equation}
\begin{equation}
    \input{figs/prove_enrich_1b.tikz}
\end{equation}
\begin{equation}
    \input{figs/prove_enrich_1c.tikz}
\end{equation}
Then doing the same to the right and side of the condition. 
\begin{equation}
    \input{figs/prove_enrich_2a.tikz}
\end{equation}
\begin{equation}
    \input{figs/prove_enrich_2b.tikz}
\end{equation}
\end{comment}
\begin{comment}
Next we confirm that the parallel composition preservation condition is satisfied, deriving equality between the following two terms, firstly:
\begin{equation}
    \input{figs/prove_enrich_3a.tikz}
\end{equation}
\begin{equation}
    \input{figs/prove_enrich_3b.tikz}
\end{equation}
and then secondly:
\begin{equation}
    \input{figs/prove_enrich_4.tikz}
\end{equation}
Finally the state-based condition can be confirmed:
\begin{equation}
    \input{figs/prove_enrich_6.tikz}
\end{equation}
\end{comment}

%%%%%%%%%%%%%%%%%%%%%%%%%%%%%%%%%%%%%%%%%%%%%%%%%%%%%%%%%%%%%%%%%%%

\section{Linked faithful categories are Closed Monoidal}
\begin{lemma}
A category $\mathbf{C}$ is closed symmetric monoidal if and only if $\mathbf{C}$ is a symmetric monoidal and furthermore
\begin{itemize}
    \item $\mathbf{C}$ is linked
    \item $\mathbf{C}$ has faithful usage
\end{itemize}
\end{lemma}
\begin{proof}
We begin by showing that the above two bullet points give closed monoidal structure. Let the three bullet points be true for $\mathbf{C}$, then to each pair $A,B \in \mathbf{C}$ assign the candidate for evaluation 
\begin{equation}
    \texttt{eval} : = \input{figs/main_thm_1_a.tikz}
\end{equation}
Since every $\circ$ is completely injective by assumption, so is every $\texttt{eval}$. Since $\eta : A \rightarrow [I,A]$ is a natural isomorphism for any $f \in \mathbf{C}(A,B)$ there exists a morphism $\kappa{f}$ such that 
\begin{equation}
    \texttt{eval} : = \input{figs/main_thm_2_a.tikz}
\end{equation}
One can apply the isomorphism $\eta$ to the partial insertion operation to generate a partial insertion using a lower level type $Y$ as opposed to the higher level type $[I,Y]$\footnote{This proof idea is also used by the authors in \cite{wilsoncausality}.}.
\begin{equation}
    \input{figs/smsc_partial_1.tikz}
\end{equation}
This partial insertion operation can be used to construct the curried version of any process $f$ from its static version $\kappa(f)$,  since 
\begin{equation}
    \input{figs/smsc_partial_2.tikz}
\end{equation}
It follows that for every process $f$ its curried version exists, that is, the co-universal arrow definition of a closed symmetric monoidal category is satisfied.

Now we demonstrate the converse. Let $\mathbf{C}$ be a closed SMC, then there exist sequential and parallel composition morphisms defined as adjuncts to circuits of evaluation morphisms. Concretely the definition of closed monoidal category enforces that there must exist processes $\otimes$ and $\circ$ satisfying, 
\begin{equation}
    \input{figs/smsc_comppar.tikz}
\end{equation}
which satisfy the coherence conditions for a symmetric monoidal category. The uniqueness property for co-universal arrows lifts to faithful usage for each sequential composition maps. Finally a monoidal natural isomorphism $A \cong [I,A]$ for the induced functor $[I,-]$ must be constructed. Indeed, up to unitor the inverse $\eta$ of $\texttt{eval}_{I \Rightarrow A}$, being an isomorphism by assumption, is such a candidate. $\texttt{eval}_{I \Rightarrow A}$ is natural for any closed monoidal category, so $\eta$ being its inverse is immediately also natural. Furthermore $\eta$ is easily checked to be monoidal. 
\begin{equation}
\input{figs/smsc_prove_monoidal.tikz}
\end{equation}
This completes the proof. 
\end{proof}

%%%%%%%%%%%%%%%%%%%%%%%%%%%%%%%%%%%%%%%%%%%%%%%%%%%%%%%%%%%%%%%%%%%%

\section{The Apex of a Merger is Closed Monoidal} 
Here we prove our main technical result. As a recap, from a series of enriched monoidal categories $\mathbf{C}_{i-1}, \mathbf{C}_{i}, \mathbf{C}_{i+1}$ a chain of \textit{raising} functors $\mathcal{R}_{i-1}^{i}: \mathbf{C}_{i-1} \longrightarrow \mathbf{C}_{i}$ can be written down with $\mathcal{R}_{i-1}^{i}(-) = [I_{i-1},-]$. 
\begin{lemma}
In any merger of infinite order, the following condition holds for the isomorphism $\mu_i^{i+1} := \mathcal{F}^{i+1}(\gamma) \circ \eta_i^{i+1}$:
\begin{equation}
    \input{figs/eta_commute.tikz},
\end{equation}
where we use $x_i$ to denote the natural isomorphisms $x_i: \mathcal{F}^i (-) \otimes \mathcal{F}^i (=) \cong \mathcal{F}^i (- \otimes =)$ which make each $\mathcal{F}^I$ strongly monoidal. 
\end{lemma}
\begin{proof}
\begin{equation}
    \input{figs/eta_commute_2b.tikz}
\end{equation} 
\end{proof}
Indeed the above property  is the key ingredient in the construction of our main result. We work with the following definition of a closed symmetric monoidal category \begin{definition}
An SMC  $\mathbf{C}$ is {\em closed} if for every $A,B \in Ob(\mathbf{C})$ there exists an object $A \Rightarrow B$ and a morphism $\texttt{eval}_{A \Rightarrow B} : A \otimes A \Rightarrow B \rightarrow B$, called the {\em evaluation morphism}, such that for all $f: A \otimes C \rightarrow B$ there exists  a unique $\bar{f}: C \rightarrow (A \Rightarrow B)$ such that $\texttt{eval}_{A \Rightarrow B} \circ (id \otimes \bar{f}) = f$. 
\end{definition}
\begin{theorem}
The apex $\mathbf{C}$ of any Merger of infinite order is a closed symmetric monoidal category
\end{theorem}

\begin{proof}

Since the coproduct $\coprod_{i} \mathcal{F}_i$ is essentially surjective, each object $A$ can be assigned an object $X_A$ an ``index" $l_A$ and an isomorphism $L_A$ such that $L_A:A \rightarrow \mathcal{F}_{l_A}(X_A)$. A compact notation can be introduced for combinations of functors of the form $[I_i,-]$. 
\begin{itemize}
    \item $\mathcal{R}_{i}^{i+1} := [I_i,-]$
    \item $\mathcal{R}_{i}^{j} := \mathcal{R}_{j-1}^{j} \circ \mathcal{R}_{i}^{j-1}$
\end{itemize}
furthermore the function $l:ob(\mathbf{C}) \rightarrow \mathbb{N}$ can be extended to lists by \[l_{AB} := \texttt{max}(l_A,l_A)\]
After which one can define the object representing the space of morphisms from $A$ to $B$ by \[A \Rightarrow B := \mathcal{F}_{l_{AB}+1}[\mathcal{R}_{l_A}^{l_{AB}}(X_A),\mathcal{R}_{l_{B}}^{l_{AB}}(X_B)]\]
This is the object representing the lifting of both $A$ and $B$ in to the $C^{l_{AB}}$ which contains them both, and then using the process object in the next category $C^{l_{AB}+1}$ to represent the processes between them. For each $A$,$B$ an evaluation $e_{A \Rightarrow B}: A \otimes (A \Rightarrow B) \rightarrow B$ can be defined by 

\begin{equation}
    \input{figs/hopt_eval_1f.tikz},
\end{equation}
where $\eta_{k}^{n}$ is defined inductively. For each $n \geq k$ then $(\eta_{k}^{n+1})_{X} := (\eta_{n}^{n+1})_{\mathcal{R}_{k}^{n} X}  \circ (\eta_{k}^{n})_{X}$, and for each $n < k$ then $\eta_k^n := (\eta_{n}^k)^{-1}$. For $\mathbf{C}$ to be closed monoidal one must show that for every $A,B,C$ and for every $f:A \otimes C \rightarrow B$ there exists a unique $\bar{f}:C \rightarrow (A \Rightarrow B)$ such that, 
\begin{equation}
    \input{figs/hopt_eval_3.tikz}
\end{equation}

Indeed such a map $\bar{f}$ can be constructed. Firstly defining $g$ such that 
\begin{equation}
    \input{figs/hopt_eval_7f.tikz}
\end{equation}
Such a $g$ must exist since each functor $\mathcal{F}_i$ is full. In terms of this $g$ define $\bar{f}$ by
\begin{equation}
    \input{figs/hopt_eval_4.tikz}
\end{equation}

Then to prove the required identity first requires repeated application of  lemma (16),
\begin{equation}
    \input{figs/hopt_eval_2f.tikz}
\end{equation}
and then using the defining identity for the partial insertion operation $\Delta$.
\begin{equation}
    \input{figs/hopt_eval_5f.tikz}
\end{equation}
and finally using monoidal naturality of the transformation $\eta_{l_{ABC}+1}^{l_{{ABC}}}$. 
\begin{equation}
    \input{figs/hopt_eval_6f.tikz}
\end{equation}
The morphism $\bar{f}$ satisfying $e \circ (A \otimes \bar{f}) = f$ must be demonstrated to be unique. Every $\mu_{i}^{j}$ is an isomorphism by fully faithful-ness of the sequence of enriched monoidal categories, as a result every morphism $h: C \rightarrow A \Rightarrow B$ can be written in the form
\begin{equation}
    \input{figs/hopt_eval_8.tikz}
\end{equation}
Where in the last line fullness of each $\mathcal{F}_i$ is used. Assuming $h$ and $h'$ have decomposition in terms of $m$ and $m'$ respectively both evaluate to the same morphism $e \circ (A \otimes h) = e \circ (A \otimes h')$:
\begin{equation}
    \input{figs/hopt_eval_9f.tikz},
\end{equation}
which in turn implies 
\begin{equation}
    \input{figs/hopt_eval_10f.tikz}.
\end{equation}
Since each $\eta$ and $L$ is an isomorphism, and each composition morphism $\bigcirc$ is part of the structure of a faithful monoidal enrichment, and each $\mathcal{F}_i$ is faithful this entails that $m = m'$ and as a result that $\bar{f} = \bar{f}'$. It follows that $\bar{f}$ is the unique morphism satisfying the evaluation condition for $f$.

\end{proof}

\end{document}